\documentclass[A4paper,12pt]{article} 

\usepackage{amsmath,amsthm,amsfonts,amssymb}
\usepackage{amsbsy,latexsym,amsopn,amstext,amsxtra,euscript,amscd}
\usepackage{array}
\usepackage{ifthen}
\usepackage{url}
\usepackage{graphicx}

\def\8u{\infty}

\newcommand{\dm}{\begin{displaymath}}
\newcommand{\md}{\end{displaymath}}
\newcommand{\ct}{\begin{center}}
\newcommand{\tc}{\end{center}}

\numberwithin{equation}{section} 

\newtheorem{theorem}{Theorem}
\newtheorem{definition}{Definition}
\newtheorem{lemma}{Lemma}
\newtheorem{remark}{Remark}

\begin{document}

\title{{\bf  Uncertainty Principle and Sparse Reconstruction in Pairs of Orthonormal Rational Function Bases}}
\author{
     \quad\ {\bf Dan Xiong$^1$, Li Chai$^2$, Jingxin Zhang$^3$ }\\
{\small 1. School of Information Science and Engineering} \\
{\small 2. Engineering Research Center of Metallurgical Automation and Measurement Technology}\\
{\small Wuhan University of Science and Technology, Hubei, Wuhan, 430081, China}\\
{\small 3. School of Software and Electrical Engineering}\\
{\small Swinburne University of Technology, Melbourne, VIC3122, Australia}}
\date{} 
\maketitle
\pagestyle{empty}  
\thispagestyle{empty} 

{\hspace{0mm}\bf {Abstract}}:\quad Most rational systems can be described in terms of orthonormal basis functions.
This paper considers the reconstruction of a sparse coefficient vector for a rational transfer function
under a pair of orthonormal rational function bases and from a limited number of linear frequency-domain measurements.
We prove the uncertainty principle concerning pairs of compressible representation of orthonormal rational functions
in the infinite dimensional function space.
The uniqueness of compressible representation using such pairs is provided as a direct consequence of uncertainty principle.
The bound of the number of measurements which guarantees the replacement of $1_0$ optimization searching for the unique sparse reconstruction by $1_1$ optimization using random sampling on the unit circle
with high probability is provided as well.

\indent{\hspace{0mm}\bf Key words}:\quad sparse system representation;
pairs of orthonormal rational function bases; uncertainty principle; $l_1$ optimization.


\section{Introduction}

{\hspace{5mm} A number of signal representations have been developed for the variety of the basis,
such as sinusoids, wavelets [1], Wilson bases [2], ridgelets [3] and curvelets [4].
The object of interest is the effective representation which requires very few significant coefficients.
However, different basis shows its own advantage in different aspect when representing the signal effectively.
For example, wavelets perform relatively poorly on high-frequency sinusoids,
for which sinusoids are very effective.
On the other hand, sinusoids perform poorly on impulsive events, for which wavelets are very effective.
Then a natural question arises: can we get a much shorter (sparser) representation using terms from each of several different bases. This question is posed by Donoho and his coworker, which has led to the use of dictionaries made from a concatenation of several orthonormal bases, and to seek representations of a signal $S$ as
\begin{equation} \label{repre}
S=\sum_{i=1}^n \gamma_i^{\phi} \phi_i+\sum_{i=1}^n \gamma_i ^{\psi}\psi_i=[\Phi \quad \Psi] \gamma,
\end{equation}
where
$
\Phi=\{\phi_1, \phi_2, \cdots, \phi_n\}
$
and
$
\Psi=\{\psi_1, \psi_2, \cdots, \psi_n\}
$
are two orthonormal bases for $\mathbb{R}^n$ with $n$ orthogonal vectors of unit length,
and
$\gamma^T= [\gamma_1^{\phi}, \gamma_2^{\phi}, \cdots, \gamma_n^{\phi}, \gamma_1^{\psi}, \gamma_2^{\psi}, \cdots, \gamma_n^{\psi}]
\in \mathbb{R}^{2n}$
is the representation coefficient.

The uniqueness of the sparse representation and the solution of the unique representation are two key problems.
Notice that (\ref{repre}) is an underdetermined set of $n$ equations with $2n$ unknowns,
the unique representation needs additional requirement, that is the sparsity.
Hence we need to minimize the number of the nonzero element in $\gamma$.
This can be expressed as an optimization problem
$$
(P_0): \quad \min_{\gamma} \|\gamma\|_0 \quad \mbox{s.t.} \quad S=[\Phi \quad \Psi] \gamma.
$$

The results have been exploited by Donoho and Huo in [5]:
if $S$ is representable as a highly sparse superposition of atoms from time¨Cfrequency dictionary
($\Phi$ is the spike basis and $\Psi$ is the Fourier basis),
then there is only one such highly sparse representation of $S$.
Specifically, the uniqueness of the solution to the $(P_0)$ problem is ensured for
$\|\gamma\|_0<\frac{1}{2}(1+\frac{1}{M})$,
where
$M=\sup_{1\leq i, j \leq n} |\langle \phi_i, \psi_j \rangle|$,
and the solution can be obtained by solving the convex optimization problem to minimize
the 1-norm of the coefficients among all decompositions:
$$
(P_1): \quad \min_{\gamma} \|\gamma\|_1 \quad \mbox{s.t.} \quad S=[\Phi \quad \Psi] \gamma.
$$

Underlying this result is a general uncertainty principle
which states that if two orthonormal bases are mutually incoherent,
no nonzero signal can have a sparse representation in both bases simultaneously.
In mathematical terminology,
if signal $S$ is expressed in each basis respectively
$$
S=\sum_{i=1}^n \alpha_i \phi_i=\sum_{i=1}^n \beta_i \psi_i,
$$
then we have
$$
\|\alpha\|_0+\|\beta\|_0 \geq 1+M^{-1},
$$
where $\alpha:=[\alpha_1, \alpha_2, \cdots, \alpha_n]^T$ and $\beta:=[\beta_1, \beta_2, \cdots, \beta_n]^T$.

The uncertainty principle can be dated back to the presentation
in the setting of discrete sequences, continuous and discrete-time functions,
and for several measures of ``concentration" (e.g., $L_2$ and $L_1$ measures)[6].
A very general uncertainty principle for operators on Banach spaces is given in [7],
including uncertainty principles for Bessel sequences in Hilbert spaces and for integral operators between measure spaces.
In the setting of pair of orthonormal bases,
uncertainty principle holds for a variety of interesting basis pairs,
not just sinusoids and spikes [5].
Related phenomena hold for functions of a real variable,
with basis pairs such as sinusoids and wavelets,
and for functions of two variables, with basis pairs such as wavelets and ridgelets.
In these settings, if a function is representable by a sufficiently sparse
superposition of terms taken from both bases,
then there is only one such sparse representation,
and it may be obtained by minimum 1-norm atomic decomposition.

In the field of signal and systems,
since the uncertainty principle leads to the uniqueness of a sparse solution
and provides the theoretical foundation for reconstruction methods,
it plays an important role in compressed sensing (CS) [8-11],
which is a new framework for simultaneous sampling and compression of signals
and has drawn much attention since its advent several years ago.
It is well known from the CS literature that a signal may have a much sparser representation in an
overcomplete basis (redundant dictionary) consisting of concatenated orthogonal bases [5, 12-16].

In the context of finite dimensional vector spaces,
[12], [13] and [13] have presented and analyzed the sparse representation
of vector signals under a pair of orthonormal bases.
In comparison to the result of [5],
[12] has presented an improved uncertainty principle with better bounds
yielding uniqueness of the $(P_0)$ solution.
The bound ensuring the uniqueness of the $(P_1)$ solution is achieved for $\|\gamma\|_0<\frac{1}{M}$,
which enlarges the class of signals whose optimal sparse representation can be found by
a simple, linear-programming based search.

As analyzed in [12] and [14], in the finite dimensional vector space,
the uncertainty principle concerning pairs of representations of $\mathbb{R}^N$ vectors in different orthonormal bases
has a direct impact on the uniqueness property of the sparse representation of such vectors using pairs of orthonormal bases
as overcomplete dictionaries.
Hence the uncertainty principle and uniqueness are fundamentally instrumental to the sparse representation of signals
under the pairs of orthonormal bases.

Rational functions are widely used in signal processing and control
to model both signals and dynamic systems.
This paper investigates the sparse representation of a rational transfer function
under the pairs of orthonormal rational function bases
and the reconstruction of the sparse coefficient for the rational transfer function
under such pairs.
Reconstruction using the generalized orthogonal basis function (GOBF) was
introduced into system identification since the work of [17].
Rational trasnfer functions shows the advantages in improving efficiency
of the representation of linear systems,
and orthonormality leads to a great simplification of
the analysis and synthesis involved in using the basis functions.
Identification and control of linear stable dynamic systems using orthogonal rational functions (ORFs)
have been widely studied over the last twenty years, see for instance
[17-31].
In the context of signal processing,
the class of rational orthonormal basis functions has turned out to be particularly useful.
These bases and the transformations that are related to them,
such as the rational orthonormal filter structures
introduced in the 1950's by Kautz, Huggins and Young [32], [33],
have proved to be instrumental to the analysis of several problems
arising from the signal and systems theory.
Orthonormal bases are useful tools in many other branches of science as well.
For example, in the mathematical literature,
the rational orthonormal functions are
utilized in rational approximations and interpolations,
which are proved to be interconnected with the least-squares problem
and have been assembled and further developed by Walsh [34].

Recently, combined with the theory of compressed sensing,
sparse system identification using ORFs has been investigated [35-39].
In the context of finite dimensional function spaces, [40] has discussed the random sampling
in the bounded orthonormal systems with one orthonormal basis from the perspective of
structured random matrix.
However, the essence of these methods
is to find a sparse representation of the system under a
single ORF basis, which may not yield the sparsest solution.

Different from the aforementioned work, we first extend the results of [12] and [14] to infinite dimensional function space
to establish the uncertainty principle and uniqueness of compressible representation for rational transfer functions,
using the uniform bound of maximal absolute inner product of the pair of the ORFs as an index.
We then derive a compressed sensing formulation for finding a sparse representation of a rational transfer function
in pairs of ORF bases using the frequency domain measurement sampled on the unit circle.
We give a lower bound of the number of measurements which guarantees the replacement of the $l_0$ optimization
by the $l_1$ optimization under a pair of ORF bases.

The contributions of this paper are:

$\bullet$ Analysis on the uncertainty principle and uniqueness for compressible representation of rational transfer functions
in pairs of ORF bases, which extends the results of [12] to infinite dimensional function space.

$\bullet$ A novel reconstruction method for rational transfer functions with finite-order combination of two ORF bases.

$\bullet$ Analysis on the lower bound of the number of measurements which guarantees exact reconstruction by the $l_1$ optimization
under a pair of ORF bases.

The rest of this paper is organized as follows.
In Section II, the uncertainty principle and uniqueness for compressible representation of rational transfer functions are given.
The sparse system reconstruction using two orthonormal bases is given in Section III.
Section IV presents the lower bound of the number of measurements which guarantees
the replacement of the $l_0$ optimization by the $l_1$ optimization,
the proof of which is given in Section V.
Section VI concludes the paper and discussed the future work.
All the omitted proofs are presented in Appendix A - C.

\section{Uncertainty Principle and Uniqueness for compressible representation of transfer functions}
{\hspace{5mm} Transfer functions of LTI systems can be put in a Hilbert space framework.
The Hardy space $H_2$ is a Hilbert space with the inner product between two rational functions
$F(z)$ and $G(z)$ defined as
\begin{equation} \label{innerproduct}
\langle F(z), G(z) \rangle=\frac{1}{2\pi i} \oint_\mathbb{T} F(z)\overline{G(z)}\frac{dz}{z}
=\frac{1}{2\pi} \int_0^{2\pi} F(e^{i\omega})\overline{G(e^{i\omega})}d \omega,
\end{equation}
where $\mathbb{T}=\{z|\,|z|=1\}.$
And the corresponding norm in the space $H_2$ is defined as
\begin{equation} \label{norm}
\|F(z)\|_{H_2}=\sqrt{\langle F(z), F(z) \rangle},
\end{equation}
which is simply denoted by $\|F(z)\|$ hereafter.
In this paper, the space of proper, stable, real-rational transfer functions is of interest
and is referred to as $RH_{2}$, a subspace of $H_{2}$.

Given a transfer function $H(z) \in RH_2$,
it has a unique representation in every ORF basis of this space.
If $\{\phi_k(z)\}_{k=1}^\infty$ is an ORF basis,
then we have
$$
H(z)=\sum_{k=1}^\infty \alpha_k \phi_k(z),
$$
where $\alpha_k=\langle H(z), \phi_k(z)\rangle$.

Suppose we have two different ORF bases $\{\phi_k(z)\}_{k=1}^\infty$, $\{\psi_l(z)\}_{l=1}^\infty$ in the $RH_2$ space.
Then every transfer function has a unique representation under the two bases respectively,
denoted as
\begin{equation} \label{H(z)}
H(z)=\sum_{k=1}^\infty \alpha_k \phi_k(z)=\sum_{l=1}^\infty \beta_l \psi_l(z).
\end{equation}
\noindent
Obviously, if $H(z)$ has a stable pole,
then the representation of $H(z)$ using impulse response is infinite.
According to classical sampling theory,
a large number of sampling data are required to guarantee the approximation performance.
However, if the pole of the selected rational bases is exactly the same as that of $H(z)$,
then the representation will be much shorter.
A natural question arises now: can we get a much shorter (sparser) representation of
$H(z)$ in a joint, overcomplete set of general rational baese, say
$$
\{\Phi(z),\Psi(z)\}=\{\phi_1(z), \phi_2(z), \cdots, \psi_1(z), \psi_2(z),\cdots\},
$$
and reconstruct the transfer function with fewer observations?

The following two theorems are established to ensure the uniqueness property
of the compressible representation of transfer function in space $H_{2}$ using pairs of ORF bases.
As the compressible representation is of interest,
we assume that the coefficients of the transfer functions discussed in this paper
form a sequence in $l_1$, that is
\begin{equation} \label{l1bounded}
0\leq \|\alpha\|_1=\sum_{k=1}^{\infty} |\alpha_k|< \infty.
\end{equation}

The sparsity defined in [12] cannot be used for rational transfer functions,
which have infinite impulse coefficients.
We first present a new definition of sparsity, called $\varepsilon$-sparsity
and then establish the uncertainty principle
that leads to the bound yielding uniqueness of the sparse representation.

\begin{definition}
For a fixed threshold $\varepsilon >0$ and an infinite sequence $\alpha=(\alpha_1, \alpha_2, \cdots)^T$ in $l_1$,
let
$$
N_{\varepsilon}(\alpha)=\min \{K: \sum_{k=K}^{\infty}|\alpha_k|\leq \varepsilon\},
$$
the $\varepsilon$-support of $\alpha$ is defined as
$$\Gamma_{\varepsilon}(\alpha)=\{k : |\alpha_k|\neq 0, 1\leq k < N_{\varepsilon}(\alpha)\},$$
and the cardinality of $\Gamma_{\varepsilon}(\alpha)$ as the $\varepsilon$-0 norm of $\alpha$,
denoted by $\| \alpha \|_{0(\varepsilon)}$.
\end{definition}

\begin{remark}
Equation (\ref{l1bounded}) guarantees the existence of $N_{\varepsilon}(\alpha)$.
If $\varepsilon=0$, then $\varepsilon$-support $\Gamma_{\varepsilon}(\alpha)$ for $\alpha$ is the support of $\alpha$ in the general sense, i.e. $\{k:|\alpha_k|\neq 0\}$.
\end{remark}

\begin{definition}
For a given positive integer s, the coefficient $\alpha$ is $(\varepsilon, s)$-sparse in the sense of $\varepsilon$-0 norm
if $\| \alpha \|_{0(\varepsilon)} \leq s$.
For brevity, we call the coefficient $\alpha$  $\varepsilon$-sparse if the value of $s$ is not concerned.
\end{definition}

\begin{definition}
A rational transfer function is $(\varepsilon, s)$-sparse
if the representation coefficient under a orthonormal rational function basis is $(\varepsilon, s)$-sparse.
\end{definition}

If $\varepsilon=0$, then $\| \alpha \|_{0(\varepsilon)}=\| \alpha \|_0$,
which is the number of nonzeros in $\alpha$,
and the $(\varepsilon, s)$-sparsity becomes the s-sparsity in compressed sensing.
However it should be noted that for a rational transfer function with poles away from zero,
its impulse response is usually infinite
and cannot be $(0, s)$-sparse corresponding to $\varepsilon=0$,
which shows that the definition of the sparsity in traditional compressed sensing is not
applicable to the sparsity of the rational transfer function.
The uncertainty principle cannot be derived either.
In this following, based on the $(\varepsilon, s)$-sparsity,
we present the uncertainty principle and uniqueness for the representation of the rational transfer
functions under two ORF bases.

\begin{theorem}\label{UP}
(Uncertainty Principle)
Let $H(z)\in RH_{2}$ be a transfer function that can be represented both as
$H(z)=\sum_{k=1}^\infty \alpha_k \phi_k(z)$ and
$H(z)=\sum_{l=1}^\infty \beta_l \psi_l(z)$.
For a fixed threshold $\varepsilon >0$,
$\Gamma_{\varepsilon}(\alpha)$ and $\Gamma_{\varepsilon}(\beta)$ are the $\varepsilon$-supports of $\alpha$ and $\beta$, respectively, whose cardinalities are
$\| \alpha \|_{0(\varepsilon)}$ and $\| \beta \|_{0(\varepsilon)}$,
then for all such pairs of representation we have
$$
{(\sqrt{\| \alpha \|_{0(\varepsilon)}}+\varepsilon)}^2
+{(\sqrt{\| \beta \|_{0(\varepsilon)} }+ \varepsilon)}^2 \geq \frac{2}{\mu},
$$
where
\begin{equation} \label{Eq2.4}
\mu= \sup_{k,l} |\langle \phi_k(z), \psi_l(z) \rangle|
\end{equation}
and $\langle\phi_k(z), \psi_l(z)\rangle$ is the inner product of $\phi_k(z)$ and $\psi_l(z)$ defined in (\ref{innerproduct}).
Following the terminology of compressed sensing,
we call $\mu$ the mutual coherence of two ORF bases $\{\phi_k(z)\}$ and $\{\psi_l(z)\}$.
\end{theorem}
\noindent
{\bf Proof:} See Appendix A.

Uncertainty Principle reveals that
in the perspective of $\varepsilon$-0 norm,
a transfer function $H(z)$ having a sparse representation in the joint set of two ORF bases
will have highly nonsparse representation in either of these bases alone.
And the uncertainty principle directly determines the bound
which guarantees the uniqueness of the sparse representation in a pair of ORF bases.
We give a new sparsity definition which is used to build the uniqueness theorem.

\begin{definition}
A rational transfer function $H(z)$ is $(\varepsilon, s)$-sparse in the pairs of orthonormal rational function bases
if the representation of $H(z)$ under a orthonormal rational function basis,
such as
\begin{equation} \label{biorthorepre}
H(z)=\sum_{k=1}^\infty \theta_k^{\phi} \phi_k(z)+\sum_{l=1}^\infty \theta_l^{\psi} \psi_l(z)
\end{equation}
with the coefficient satisfying $(\varepsilon, s)$-sparsity.
For brevity, denote $\theta_1=(\theta_1^{\phi}, \theta_2^{\phi}, \cdots)$
and $\theta_2=(\theta_1^{\psi}, \theta_2^{\psi}, \cdots)$,
so the $(\varepsilon, s)$-sparsity for the rational transfer function $H(z)$ is equivalence to
$\|\theta_1\|_{0(\varepsilon)}+ \|\theta_2\|_{0(\varepsilon)} \leq s$.
\end{definition}

\begin{theorem}\label{Unique}
(Uniqueness)
For a fixed threshold $\varepsilon >0$,
assume $H(z)$ is $(\varepsilon, s)$-sparse under a pair of ORF bases
$\{\phi_k(z)\}_{k=1}^\infty$ and $\{\psi_l(z)\}_{l=1}^\infty$.
Then the representation (\ref{biorthorepre}) is unique if
$$
{(\sqrt{\|\theta_1\|_{0(\varepsilon)} }
+\varepsilon)}^2+{(\sqrt{\|\theta_2\|_{0(\varepsilon)} }+ \varepsilon)}^2
< \frac{1}{\mu},
$$
where $\theta_1$ and $\theta_2$ as defined in Definition 4.
\end{theorem}
\noindent
{\bf Proof:} See Appendix A.

\begin{remark}
If $\varepsilon =0$ and the sequence is finite,
then the results of Theorem \ref{UP} and \ref{Unique} are parallel to the results in [12].
\end{remark}

The sparsity bound determined by $\mu$ is crucial.
Since the transfer function can be represented using the impulse responses,
a simple calculation formula of $\mu$ as shown below can be obtained.

\begin{lemma} \label{Lemma1}
For two arbitrary ORF bases $\{\phi_k(z)\}_{k=1}^\infty$ and $\{\psi_l(z)\}_{l=1}^\infty$,
$$\mu=\sup_{k, l} |\sum_{d=0}^{\infty} b_{dk} a_{dl}|,$$
where $\{b_{dk}\}$ and $\{a_{dl}\}$ are the impulse responses of the two bases, respectively.
\end{lemma}

\begin{proof}
The impulse responses of bases $\phi_k(z)$ and $\psi_l(z)$
can be written as
$$\phi_{k}(z):=\sum_{d'=0}^{\infty}b_{d'k}z^{-d'}, \quad k=1, 2, \cdots, $$
and
$$\psi_l(z):=\sum_{d=0}^{\infty}a_{dl}z^{-d}, \quad l=1, 2, \cdots, $$
respectively.

Then the inner product of $\phi_{k}(z)$ and $\psi_{l}(z)$ becomes
$$
\langle \phi_{k}(z), \psi_{l}(z)\rangle
= \langle\sum_{d'=0}^{\infty}b_{d'k}z^{-d'}, \sum_{d=0}^{\infty}a_{dl}z^{-d}\rangle
= \sum_{d'=0}^{\infty} \sum_{d=0}^{\infty}b_{d'k} a_{dl} \langle z^{-d'},z^{-d}\rangle
= \sum_{d=0}^{\infty} b_{dk} a_{dl},
$$
which equals the inner product of the impulse responses of the two basis functions.
With $\mu$ defined in (\ref{Eq2.4}), the claim follows.
\end{proof}

Lemma \ref{Lemma1} shows that mutual coherence $\mu$ is the uniform bound of the maximal absolute inner product of the impulse responses of the two basis functions.

\section{Sparse system reconstruction using pairs of orthonormal rational function bases}
{\hspace{5mm}
Given a pair of ORF bases $\{\phi_{k}(z)\}$ and $\{\psi_{l}(z)\}$,
assume that $H(z)$ has a $\varepsilon$-sparse representation under such pairs as in (\ref{biorthorepre}).
We want to reconstruct $H(z)$ with a small fraction of the measurements of $H(z)$ on the unit circle.
Precisely, define
$$
T_N:=\{z_{r}=e^{2\pi i (r-1)/N},\quad r=1,2,\cdots,N\}.
$$
We focus on the underdetermined case
with only a few of the components of $\{H(z_r),  r=1, 2, \cdots, N\}$ sampled or observed.
That is, only a small fraction of $T_N$ is known.
Given a subset $\Omega \subset \{1,2,\cdots,N\}$ of size $|\Omega|=m$.
The goal is to reconstruct the representation coefficient and hence the transfer function $H(z)$
from the much shorter $m$-dimensional measurements $\{H(z_r), r \in \Omega\}.$

Now we will restate this problem in a matrix form.
Combining with Definition 1,
$H(z)$ can be rewritten as
$$
H(z)
=\sum\limits_{k=1}^{n_1} \theta_k^{\phi} \phi_k(z) + \sum\limits_{l=1}^{n_2} \theta_l^{\psi} \psi_{l}(z) + \Delta_1 +\Delta_2,
$$
where $n_1=N_{\varepsilon}(\theta_1)-1$ and $n_2=N_{\varepsilon}(\theta_2)-1$,
$\Delta_1=\sum\limits_{k=n_1+1}^{\infty} \theta_k^{\phi} \phi_k(z)$
and
$\Delta_2=\sum\limits_{k=n_2+1}^{\infty} \theta_l^{\psi} \psi_l(z)$.

By simple calculation, we have
$$
\|\Delta_1\|^2
= \langle \sum\limits_{k=n_1+1}^{\infty}  \theta_k^{\phi} \phi_k(z), \sum\limits_{k=n_1+1}^{\infty} \theta_k^{\phi} \phi_k(z)\rangle
= \sum\limits_{k=n_1+1}^{\infty} \! |\theta_k^{\phi}|^2
\leq \varepsilon^2.
$$

Similarly, we have
$
\|\Delta_2\|^2 \leq \varepsilon^2.
$

Denote $\Delta=\Delta_1+\Delta_2$, then we have
$$
\|\Delta\| \leq \sqrt{2(\|\Delta_1\|^2 + \|\Delta_2\|^2)} \leq 2 \varepsilon.
$$

Now the transfer function $H(z)$ can be simplified as
\begin{equation} \label{simplever}
H(z)
=\sum\limits_{k=1}^{n_1} \theta_k^{\phi} \phi_k(z) + \sum\limits_{l=1}^{n_2} \theta_l^{\psi} \psi_{l}(z) + \Delta,
\end{equation}
with $\|\Delta\| \leq 2 \varepsilon$.
With a little bit abuse of the notation,
the unknown coefficients to be determined here are denoted as
$\theta_1=[\theta_1^{\phi}, \theta_2^{\phi}, \cdots, \theta_{n_1}^{\phi}]^T$ and
$\theta_2=[\theta_1^{\psi}, \theta_2^{\psi}, \cdots, \theta_{n_2}^{\psi}]^T$.

Since the arbitrariness of $\varepsilon$, the norm of the term $\Delta$ can be small enough.
And the term $\Delta=0$ when $\varepsilon$ is exactly zero.
In the sequel, we discuss the equation (\ref{simplever}) with the term $\Delta$ omitted.

Define
\begin{equation} \label{total sample matrix}
[\Phi \quad \Psi]:= [(\phi_k(z_r))_{k=1}^{n_1} \ (\psi_l(z_r))_{l=1}^{n_2}], \ r=1,2,\cdots,N,
\end{equation}
and
\begin{equation} \label{total measurements}
H:=[H(z_1), H(z_2), \cdots, H(z_N)]^T.
\end{equation}
\noindent
Then
\begin{equation} \label{matrix form}
H=
[\Phi \quad \Psi]
\begin{bmatrix}
\theta_1\\
\theta_2
\end{bmatrix}.
\end{equation}
\noindent

We randomly select the subset $\Omega$ of size $m (<< n_1+n_2)$ drawn from
the uniform distribution over the index set $\{1,2,\cdots,N\}$,
and denote the measurement by
\begin{equation} \label{Eq4.2}
H_{\Omega}=[\Phi \quad \Psi]_{\Omega}
\begin{bmatrix}
\theta_1\\
\theta_2
\end{bmatrix},
\end{equation}
where $H_{\Omega}$ is the $m \times 1$ vector consisting of $\{H(z_r),  r \in \Omega\}$,
and $[\Phi \quad \Psi]_{\Omega}$ is the $m \times (n_1+n_2)$ matrix $[(\phi_k(z_r))_{k=1}^{n_1} \quad (\psi_l(z_r))_{l=1}^{n_2}], r \in \Omega$.

As $[\Phi \quad \Psi]_{\Omega}$ is the concatenation of two bases, the representation is not unique.
While the uniqueness of the representation is guaranteed if the representation is sufficiently sparse.
The goal is to find the sparsest decomposition from the $l_0$ minimization
$$
(P_0): \
\min_{\theta_1, \theta_2} \left\|
\begin{bmatrix}
\theta_1\\
\theta_2
\end{bmatrix}
\right \|_{0}
\ \mbox{subject to} \quad H_{\Omega}=[\Phi \quad \Psi]_{\Omega}
\begin{bmatrix}
\theta_1\\
\theta_2
\end{bmatrix},
$$
which is an infeasible search problem [41].
An alternative approach called Basis Pursuit [8], [11], [42]
\begin{equation} \label{Eq4.3}
(P_1): \
\min_{\theta_1, \theta_2} \left \|
\begin{bmatrix}
\theta_1\\
\theta_2
\end{bmatrix}
\right \|_{1}
\ \mbox{subject to} \quad H_{\Omega}=[\Phi \quad \Psi]_{\Omega}
\begin{bmatrix}
\theta_1\\
\theta_2
\end{bmatrix}.
\end{equation}

For the sparse representation using only one basis,
compressed sensing theory has presented
the equivalence of $l_0$ optimization and $l_1$ minimization when the representation is sufficiently sparse [41], [42], and have provided the sufficient conditions on the number of measurements needed to
recover the sparse coefficient from the randomly sampled measurements
by solving the $l_1$-minimization problem [40], [41].

However, the existing results cannot be applied to the setting of (\ref{Eq4.3}), for there are two bases concerned.
Hence the sufficient conditions on the equivalence of $l_0$ optimization and $l_1$ minimization should be reconsidered.
Here we first show the orthonormality of $[\Phi \quad \Psi]$ When $N$ is sufficiently large.
With the orthonormality satisfied,
we then present a sufficient condition for exact reconstruction by $l_1$ optimization in pairs of orthonormal bases
in the next Section.
\begin{theorem} \label{ThOrho}
When $N$ is sufficiently large, the composite sampling matrix $[\Phi \quad \Psi]$ satisfies:

(1) $\Phi^{*} \Phi \approx N I_{n_1}$,
where $^*$ is the conjugate transpose, $I_{n_1}$ is the identity matrix of dimension $n_1$.

(2) $\Psi^{*} \Psi \approx N I_{n_2}$.

(3) $
\Phi^{*} \Psi
= (\sum_{r=1}^N \psi_k(z_r) \overline{z_r^{-l}})
\approx
N(\sum_{d=0}^{\infty} b_{dk} a_{dl}),
\quad (k=1, \cdots, n_1, l=1, \cdots, n_2).
$
\end{theorem}
\noindent
{\bf Proof:} See Appendix A.

\section{The lower bound of the number of measurements for exact reconstruction by $l_1$ optimization}
{\hspace{5mm}
In this section, we give the lower bound of the number of measurements
which guarantees the replacement of $l_0$ optimization $(P_0)$ by $l_1$ optimization $(P_1)$
under two ORF bases for a fixed (but arbitrary) support.
We first introduce some notations.

Let $[\Phi \quad \Psi]$ be a composite matrix satisfying
\begin{equation}\label{ortho}
[\Phi \quad \Psi]^*[\Phi \quad \Psi]
=N \begin{bmatrix}
I_{n_1} & \Phi^* \Psi/N\\
\Psi^* \Phi/N  & I_{n_2} \\
\end{bmatrix},
\end{equation}
where $ \Phi^* \Psi/N$ is a matrix with the $(i, j)$-th element
being the coherence of $i$-th column vector $\Phi_i$ of $\Phi$ and the $j$-th column vector $\Psi_j$ of $\Psi$, i.e.
$
\frac{\langle \Phi_i, \Psi_j \rangle}{\|\Phi_i\| \|\Psi_j\|}.
$
The mutual coherence of matrices $\Phi$ and $\Psi$ is defined as
$$
\mu(\Phi, \Psi)=\max_{i, j}\frac{\langle \Phi_i, \Psi_j \rangle}{\|\Phi_i\| \|\Psi_j\|}.
$$

Define $\mu_M=\max\{\mu_\Phi, \mu_\Psi\}$,
where $\mu_\Phi=\max|\Phi_{ij}|$ and $\mu_\Psi=\max|\Psi_{ij}|$
are the largest magnitude among the entries in $\Phi$ and $\Psi$, respectively.

Denote $T_1=\{k: |\theta_k^{\phi}| \neq 0\}$ and  $T_2=\{l: |\theta_l^{\psi}|\neq 0\}$
as the support of coefficient $\theta_1$ and $\theta_2$, respectively.
Let $\Phi_{T_1}$ be the $N \times |T_1|$ matrix corresponding to the columns of $\Phi$ indexed by $T_1$,
and let $\Phi_{\Omega T_1}$ be the $m \times |T_1|$ matrix corresponding to the rows of $\Phi_{T_1}$ indexed by $\Omega$.
And $\Psi_{T_2}$ and $\Psi_{\Omega T_2}$ are defined similarly.
Further let
$$
F=\begin{bmatrix}
I & \Phi_{T_1}^*\Psi_{T_2}/N \\
\Psi_{T_2}^*\Phi_{T_1}/N  &  I
\end{bmatrix}.
$$

\begin{theorem}\label{Th6}
Let $[\Phi \quad \Psi]$ be a matrix obeying (\ref{ortho}).
Fix a subset $T=T_1 \bigcup T_2$ of the coefficient domain,
with $T_1$ and $T_2$ being a subset of the coefficient domain of $\Phi$ and $\Psi$, respectively.
Choose a subset $\Omega$ of the measurement domain of size $|\Omega|=m$, and a sign sequence $\tau$ on $T$ uniformly at random. Suppose $m$ satisfies
$$
m \geq C \mu_M^2 {\max}^2 \{|T|, \log{(\frac{n_1+n_2}{\delta})}, C_{F,\mu(\Phi, \Psi),T,\delta}\},
$$
where
\[
C_{F,\mu(\Phi, \Psi),T,\delta}=\frac{4}{\left[(\frac{1}{2}+\|\Phi_{T_1}^*\Psi_{T_2} /N\|){[2 \log{(\frac{2(n_1+n_2)}{\delta})}]}^{-1/2} - \mu(\Phi, \Psi) \sqrt{|T|}\right]^2}.
\]
Then with the probability exceeding $1-6\delta$,
every coefficient vector
$\theta=\begin{bmatrix}
\theta_1\\
\theta_2
\end{bmatrix}$
supported on $T$ with sign matching $\tau$ can be recovered from solving
\begin{equation}\label{l0}
\min \|\hat{\theta}\|_1 \quad \mbox{s.t.} \quad [\Phi_\Omega \quad \Psi_\Omega] \hat{\theta}=H_{\Omega}
\end{equation}
for the coefficient vector
$\hat{\theta}=\begin{bmatrix}
\hat{\theta}_1\\
\hat{\theta}_2
\end{bmatrix}$,
where $H_{\Omega}=[\Phi_\Omega \quad \Psi_\Omega] \theta$.
\end{theorem}

Theorem \ref{Th6} shows that for most sparse coefficient $\theta$ supported on a fixed (but arbitrary) set $T$,
the coefficient can be recovered with overwhelming probability if the sign of $\theta$ on $T$ and the observations
$H_{\Omega}=[\Phi_\Omega \quad \Psi_\Omega] \theta$ are drawn at random.

\begin{remark}
(a) If $\mu$ is small enough, then $C_{F,\mu (\Phi, \Psi),T,\delta}$ equals $32 \log{(\frac{2(n_1+n_2)}{\delta})}$.
Hence $m$ in Theorem 1 is simplified to $m \geq C \mu_M^2 {\max}^2 \{|T|, \log{(\frac{n_1+n_2}{\delta})}\}$.

(b) Another important parameter in Theorem \ref{Th6} is $\mu_M$.
Since each column of $\Phi$ and $\Psi$ has an $l_2$-norm equal to $\sqrt{N}$.
$\mu_M$ take a value between 1 and $\sqrt{N}$.
When the columns of $\Phi$ and $\Psi$ are perfectly flat-
all elements has modulus equal to 1,
we will have $\mu_M=1$, and the bound of $m$ is good.
But when a column is maximally concentrated-
all the column entries but one vanish,
then $\mu_M=\sqrt{N}$,
and Theorem \ref{Th6} offers us no guarantees for recovery from a limited number of samples.
\end{remark}

\section{Proof of Theorem \ref{Th6}}
{\hspace{5mm}
To prove Theorem \ref{Th6}, we need the result essentially proposed in [11].
Here we restate it by our notations.

\begin{lemma}\label{Lemma2}
The coefficient $\theta$ defined in Theorem \ref{Th6} is the unique solution to (\ref{l0})
if and only if there exists a dual vector $\pi \in \mathbb{R}^{n_1+n_2}$ satisfying the following properties:

(1) $\pi$ is in the row space of $[\Phi_\Omega \quad \Psi_\Omega]$,

(2) $\pi(t)=\text{sgn} \theta(t)$ for $t \in T=T_1 \bigcup T_2$,

(3) $|\pi(t)|<1$ for $t \in T^c$.
\end{lemma}

In the following, we first define the dual vector $\pi$ as
\begin{equation}\label{pi}
\pi=[\Phi_\Omega \quad \Psi_\Omega]^* [\Phi_{\Omega T_1} \quad \Psi_{\Omega T_2}][[\Phi_{\Omega T_1} \quad \Psi_{\Omega T_2}]^* [\Phi_{\Omega T_1} \quad \Psi_{\Omega T_2}]]^{-1} \tau_0,
\end{equation}
where $\tau_0$ is a $|T|$-dimensional vector whose entries are the signs of $\theta$ on $T$.
Then we will show that $\pi$ in (\ref{pi}) satisfies condition $(1) - (3)$ in Lemma \ref{Lemma2}.
For this purpose,
the invertibility of $\pi$, which guarantees the definition of $\pi$,  is given in Theorem \ref{Th7},
and the conditions leading to the property (3) is given in Theorem \ref{Th8},
The theorems are presented as below and their proofs are given in the Appendix B and C, respectively.

\begin{theorem} \label{Th7}
Let $[\Phi \quad \Psi]$ be a matrix obeying (\ref{ortho}).
Consider a fixed set $T$ and let $\Omega $ be a random set sampled using the Bernoulli model.
Suppose that the number of measurements $m$ obeys
$$
m \geq |T| \mu_M^2 \cdot \max\{4C_R^2(1+3\|F\|)\log{|T|}, C_T \log{(3/\delta)} \},
$$
for some positive constants $C_R, C_T$. Then

(1)$$
P\left(
\|\frac{1}{m}
[\Phi_{\Omega T_1} \quad \Psi_{\Omega T_2}]^* [\Phi_{\Omega T_1} \quad \Psi_{\Omega T_2}]-F
\|> 1/2\right)\leq \delta,
$$
where $\|\cdot\|$ is the spectral norm, the amplitude of the largest eigenvalue.

(2)$$
P\left(\|\frac{1}{m} [\Phi_{\Omega T_1} \quad \Psi_{\Omega T_2}]^* [\Phi_{\Omega T_1} \quad \Psi_{\Omega T_2}]
- I\|> 1/2 + \|\Phi_{T_1}^*\Psi_{T_2} /N\| \right) \leq  \delta.
$$
\end{theorem}

Theorem \ref{Th7} reveals that for small value of $\delta$ and $\mu (\Phi, \Psi)$, the eigenvalues of
$[\Phi_{\Omega T_1} \quad \Psi_{\Omega T_2}]^* $ $[\Phi_{\Omega T_1} \quad\Psi_{\Omega T_2}]$
are all close to $m$ with high probability,
which is an uncertainty principle.
Let $\theta$ be a sequence supported on $T$,
then with probability exceeding $1-\delta$, we have
$
\|\frac{1}{m} [\Phi_{\Omega T_1} \quad \Psi_{\Omega T_2}]^* [\Phi_{\Omega T_1}$ \quad $\Psi_{\Omega T_2}]- I\|\leq 1/2 + \|\Phi_{T_1}^*\Psi_{T_2} /N\|
$.
It follows that
$$
\left(\frac{1}{2}+\|\Phi_{T_1}^*\Psi_{T_2} /N\|\right)m \|\theta\|^2
\leq \|(\Phi_{\Omega T_1} \quad \Psi_{\Omega T_2})\theta\|^2
\leq \left(\frac{3}{2}+\|\Phi_{T_1}^*\Psi_{T_2} /N\|\right)m \|\theta\|^2,
$$
which shows that only a small portion of the energy of $\theta$ will by concentrated on the set $\Omega$
in the $[\Phi \quad \Psi]$-domain.

\begin{theorem} \label{Th8}
Let $[\Phi \quad \Psi]$, $T$ and $\Omega $ be as defined in Theorem \ref{Th7}.
Denote
$\lambda=((\frac{1}{2}+\|\Phi_{T_1}^*\Psi_{T_2} /N\|) )^{-1}(\mu (\Phi, \Psi) \sqrt{|T|} + \frac{a\bar{\sigma}}{m}  + \frac{\sqrt{|T|} \mu_M}{\sqrt{m}})$
with
$\bar{\sigma}^2= m \mu_M^2 \max\{ 2 , \frac{|T| \mu_M}{\sqrt{m}}\}$.
For each $a>0$ obeying $a \leq \frac{ \sqrt{2 m}}{\sqrt{|T|}\mu_M}$ if $\frac{|T| \mu_M}{\sqrt{m}} \leq 2$
and $a \leq ({\frac{m}{\mu_M^2}})^{1/4}$ otherwise, then
\begin{eqnarray*}
P(\sup_{t \in T^c} |\pi(t)| \geq 1 )
&\leq& 2 (n_1+n_2) e^{-\frac{1}{2 \lambda^2}}
+ 3(n_1+n_2) e^{-\gamma a^2}
\\
&+& P\left(\|[\Phi_{\Omega T_1} \quad \Psi_{\Omega T_2}]^* [\Phi_{\Omega T_1} \quad \Psi_{\Omega T_2}]\|
\leq (\frac{1}{2}+\|\Phi_{T_1}^*\Psi_{T_2} /N\|) m \right)
\end{eqnarray*}
\end {theorem}

\noindent
{\bf Proof of Theorem \ref{Th6}.}

It is obvious that properties (1) and (2) in Lemma \ref{Lemma2} hold.
By Theorem \ref{Th7},
we know that the invertibility of $[\Phi_{\Omega T_1} \quad \Psi_{\Omega T_2}]^* [\Phi_{\Omega T_1} \quad \Psi_{\Omega T_2}]$
is guaranteed with high probability.
Now we will show the property (3) holds with high probability.

Set $\lambda$ and $a>0$ satisfying the condition in Theorem \ref{Th8}, we have
\begin{eqnarray*}
P(\sup_{t \in T^c} |\pi(t)| \geq 1 )
&\leq& 2 (n_1+n_2) e^{-\frac{1}{2 \lambda^2}}
+ 3(n_1+n_2) e^{-\gamma a^2}  \\
&+& P\left(\|[\Phi_{\Omega T_1} \quad \Psi_{\Omega T_2}]^* [\Phi_{\Omega T_1} \quad \Psi_{\Omega T_2}]\|
\leq (\frac{1}{2}+\|\Phi_{T_1}^*\Psi_{T_2} /N\|) m \right)
\end{eqnarray*}

For the second term to be less than $\delta$, we choose $a$ such that
\begin{equation}\label{a}
a^2=\gamma^{-1}\log{(\frac{3(n_1+n_2)}{\delta})}.
\end{equation}
\noindent
The first term is less than $\delta$ if
\begin{equation}\label{lambda}
\frac{1}{\lambda^2} \geq 2 \log{(\frac{2(n_1+n_2)}{\delta})}.
\end{equation}
\noindent
When $\frac{|T| \mu_M}{\sqrt{m}} >2$, the condition in Theorem \ref{Th8} is $a \leq ({\frac{m}{\mu_M^2}})^{1/4}$.
Combining with (\ref{a})
$$
a^2=\gamma^{-1}\log{(\frac{3(n_1+n_2)}{\delta})} \leq ({\frac{m}{\mu_M^2}})^{1/2},
$$
equivalently
\begin{equation} \label{m1}
m \geq \mu_M^2 \gamma^{-2} [\log{(3(n_1+n_2)/\delta)}]^2.
\end{equation}
\noindent
In this case, $a \bar{\sigma} \leq \sqrt{m|T|} \mu_M$, then
\begin{equation} \label{lambda1}
\lambda
\leq (\frac{1}{2}+\|\Phi_{T_1}^*\Psi_{T_2} /N\|)^{-1}(\mu (\Phi, \Psi) \sqrt{|T|} +2 \mu_M \sqrt{\frac{|T|}{m}}).
\end{equation}

When $ \frac{|T| \mu_M}{\sqrt{m}} \leq 2$, the condition in Theorem \ref{Th8} is $a \leq \frac{\sqrt{2 m}}{\sqrt{|T|}\mu_M}$.
That is, when
\begin{equation} \label{m2}
m \geq \frac{|T|^2 \mu_M^2}{4},
\end{equation}
with equation (\ref{a})
$$
a^2=\gamma^{-1}\log{(\frac{3(n_1+n_2)}{\delta})}\leq \frac{2m}{|T|\mu_M^2},
$$
or equivalently
\begin{equation} \label{m3}
m \geq (2\gamma)^{-1}\log{(\frac{3(n_1+n_2)}{\delta})} |T|\mu_M^2.
\end{equation}

And if $|T| \geq 2a^2$, then $a\bar{\sigma} \leq \sqrt{\frac{T}{2}} \sqrt{2 m} \mu_M= \sqrt{m|T|} \mu_M$,
which gives again (\ref{lambda1}).

On the other hand, if $|T| \leq 2a^2$,
then $\sqrt{\frac{|T|}{m}} \mu_M \leq\frac{\sqrt{2}a}{\sqrt{m}} \mu_M =\frac{a\bar{\sigma}}{m}$,
which gives
\begin{eqnarray}
\lambda
&\leq& (\frac{1}{2}+\|\Phi_{T_1}^*\Psi_{T_2} /N\|)^{-1}(\mu (\Phi, \Psi) \sqrt{|T|} +2 \frac{a\bar{\sigma}}{m}) \nonumber \\
&=& (\frac{1}{2}+\|\Phi_{T_1}^*\Psi_{T_2} /N\|)^{-1}(\mu (\Phi, \Psi) \sqrt{|T|} + 2\sqrt{2} \frac{a \mu_M}{\sqrt{m}})\label{lambda2}
\end{eqnarray}

From (\ref{lambda1}) and (\ref{lambda2}),
$$
\lambda (\frac{1}{2}+\|\Phi_{T_1}^*\Psi_{T_2} /N\|)
\leq \mu (\Phi, \Psi) \sqrt{|T|} + 2 \frac{\mu_M}{\sqrt{m}} \max\{\sqrt{|T|}, \sqrt{2} a\}.
$$
To verify (\ref{lambda}),
it suffices to take $m$ obeying
$$
\mu (\Phi, \Psi) \sqrt{|T|} +  2 \frac{\mu_M}{\sqrt{m}} \max\{\sqrt{|T|}, \sqrt{2} a\}
\leq
(\frac{1}{2}+\|\Phi_{T_1}^*\Psi_{T_2} /N\|)
[2 \log{(\frac{(2(n_1+n_2)}{\delta}})]^{-1/2},
$$
which is equivalent to
\begin{eqnarray}
m &\geq& \left[\frac{2 \mu_M \max\{\sqrt{|T|}, \sqrt{2} a\}}
{(\frac{1}{2}+\|\Phi_{T_1}^*\Psi_{T_2} /N\|){[2 \log{(\frac{2(n_1+n_2)}{\delta})}]}^{-1/2} - \mu \sqrt{|T|}}\right]^2 \nonumber\\
&=&\frac{4 \mu_M^2 \max\{|T|, 2 a^2\}}
{\left[(\frac{1}{2}+\|\Phi_{T_1}^*\Psi_{T_2} /N\|){[2 \log{(\frac{2(n_1+n_2)}{\delta})}]}^{-1/2} - \mu (\Phi, \Psi)\sqrt{|T|}\right]^2} \nonumber\\
&=&\frac{4 \mu_M^2 \max\{|T|, 2 \gamma^{-1}\log{(\frac{3(n_1+n_2)}{\delta})}\}}
{\left[(\frac{1}{2}+\|\Phi_{T_1}^*\Psi_{T_2} /N\|){[2 \log{(\frac{2(n_1+n_2)}{\delta})}]}^{-1/2} - \mu (\Phi, \Psi) \sqrt{|T|}\right]^2}. \label{m4}
\end{eqnarray}

This analysis shows that the second term is less than $\delta$ if
$m$ satisfies (\ref{m1}), (\ref{m2}), (\ref{m3}) and (\ref{m4}), which can be simplified as
$$
m \geq K_1 \mu_M^2 {\max}^2 \{|T|, \log{(\frac{n_1+n_2}{\delta})}, C_{F,\mu (\Phi, \Psi),T,\delta}\},
$$
where
$C_{F,\mu (\Phi, \Psi),T,\delta}=\frac{4}{\left[(\frac{1}{2}+\|\Phi_{T_1}^*\Psi_{T_2} /N\|){[2 \log{(\frac{2(n_1+n_2)}{\delta})}]}^{-1/2} - \mu (\Phi, \Psi) \sqrt{|T|}\right]^2}$.

Finally, by Theorem \ref{Th7}, the last term will be bounded by $\delta$ if
$$
m \geq K_2 |T| \mu_M^2 \log{(\frac{n_1+n_2}{\delta})}.
$$

In conclusion, when using Bernoulli model the reconstruction is exact with probability at least $1-3\delta$
provided that the number of measurements $m$ satisfies
$$
m \geq K_3 \mu_M^2 {\max}^2 \{|T|, \log{(\frac{n_1+n_2}{\delta})}, C_{F,\mu (\Phi, \Psi),T,\delta}\}.
$$

Following [11],
the desired properties hold with $\Omega$ sampled using uniform model
whenever the desired properties hold when $\Omega$ is sampled using a Bernoulli model.
In fact, suppose $\Omega_1$ of size $m$ is sampled uniformly at random and $\Omega_2$ is sampled by setting
$$\Omega_2=\{k:\delta_k=1\},$$
where $\{\delta_k\}$ is a sequence of independent identically distributed $0/1$ Bernoulli random variables with probability
\begin{equation} \label{delta}
P(\delta_k=1)=\frac{m}{N}.
\end{equation}
Then
$$
P(Failure(\Omega_1))\leq 2 P(Failure(\Omega_2)).
$$
Hence, for $\Omega$ sampled using the uniform model,
the existence of a dual vector for $\theta^0$ is guaranteed with probability exceeding $1-6\delta$.
The theorem is proved.
$\hfill{} \Box$

\section{Conclusion}
{\hspace{5mm}
Based on the principle of compressed sensing,
we have proposed a reconstruction method for a sparse rational transfer function under two ORF bases.
We have established the uncertainty principle concerning compressible representation of rational functions under two ORF bases,
and we also presented the uniqueness of compressible representation.
The lower bound of the number of measurements which guarantees the replacement of $1_0$ optimization searching for the unique sparse reconstruction with random sampling on the unit circle by $1_1$ optimization with high probability is provided as well.

Since the signal and systems can both be represented by rational functions,
the proposed reconstruction method can be applied in the sparse reconstruction for signals and sparse system identification.
The linearity of parameter in the representation
shows that the reconstruction method
can be applied to sparse rational MIMO system as well.

\appendix
\section{Proof of Theorem \ref{UP}, \ref{Unique} and \ref{ThOrho}}

\noindent
\textit{\textbf{Proof of Theorem \ref{UP}.}}
Notice that the transfer function considered is in $RH_{2}$ space,
satisfying
$$
H(z)=\sum_{k=1}^\infty \alpha_k \phi_k(z)=\sum_{l=1}^\infty \beta_l \psi_l(z).
$$
For $\{\alpha_k\}$ and $\{\beta_l\}$,
denote $\Gamma_{\varepsilon}(\alpha)$ and $\Gamma_{\varepsilon}(\beta)$
as the $\varepsilon$-support of $\alpha$ and $\beta$, respectively.
Then we have
$\sum_{k \notin \Gamma_{\varepsilon}(\alpha)} \leq \varepsilon$
and
$\sum_{l \notin \Gamma_{\varepsilon}(\beta)} \leq \varepsilon$.

Without loss of generality, we assume $\langle H(z), H(z)\rangle=1$, then
$$1=\langle H(z), H(z)\rangle=\langle\sum_{k=1}^\infty \alpha_k \phi_k(z), \sum_{l=1}^\infty \beta_l \psi_l(z)\rangle=\sum_{k=1}^\infty \sum_{l=1}^\infty \alpha_k \langle \phi_k(z), \psi_l(z)\rangle \beta_l.$$
Note that
\begin{eqnarray*}
1&=&|\langle H(z), H(z)\rangle|
=|\sum_{k=1}^\infty \sum_{l=1}^\infty \alpha_k \langle \phi_k(z), \psi_l(z)\rangle \beta_l|\\
&\leq&  \mu \sum_{k=1}^\infty \sum_{l=1}^\infty |\alpha_k| |\beta_l|
= \mu \sum_{k=1}^\infty |\alpha_k|  \sum_{l=1}^\infty |\beta_l| \\
&=& \mu (\sum_{k \in \Gamma_{\varepsilon}(\alpha)} |\alpha_k|+ \sum_{k \notin \Gamma_{\varepsilon}(\alpha)}|\alpha_k|)
(\sum_{l \in \Gamma_{\varepsilon}(\beta)} |\beta_l|+\sum_{l \notin \Gamma_{\varepsilon}(\beta)} |\beta_l|)\\
&\leq&
\mu (\sum_{k \in \Gamma_{\varepsilon}(\alpha)}|\alpha_k| +\varepsilon)
(\sum_{l \in \Gamma_{\varepsilon}(\beta)} |\beta_l|+\varepsilon).
\end{eqnarray*}

Similarly, we have
\begin{eqnarray*}
1&=&\langle H(z), H(z)\rangle
=\langle\sum_{k=1}^\infty \alpha_k \phi_k(z), \sum_{l=1}^\infty \alpha_l \phi_l(z)\rangle\\
&=&\sum_{k=1}^\infty \sum_{l=1}^\infty \alpha_k \langle \phi_k(z), \phi_l(z)\rangle \alpha_l
=\sum_{k=1}^\infty \alpha_k \langle \phi_k(z), \phi_k(z)\rangle \alpha_k\\
&=&\sum_{k=1}^\infty |\alpha_k|^2= \sum_{k \in \Gamma_{\varepsilon}(\alpha)}  |\alpha_k|^2 + \sum_{k \notin \Gamma_{\varepsilon}(\alpha)}  |\alpha_k|^2\\
&\leq &  \sum_{k \in \Gamma_{\varepsilon}(\alpha)}  |\alpha_k|^2 + \sum_{k \notin \Gamma_{\varepsilon}(\alpha)}  |\alpha_k| \varepsilon\\
&\leq &  \sum_{k \in \Gamma_{\varepsilon}(\alpha)}  |\alpha_k|^2 + \varepsilon^2
\end{eqnarray*}
and
$$1=\sum_{l \in \Gamma_{\varepsilon}(\beta)}  |\beta_l|^2 +\sum_{l \notin \Gamma_{\varepsilon}(\beta)}  |\beta_l|^2
\leq \sum_{l \in \Gamma_{\varepsilon}(\beta)}  |\beta_l|^2 + \varepsilon^2
.$$
The bound of the above expression can be solved by the optimization problem
\begin{equation} \label{Eq7.1}
\max_{\alpha_k, \beta_l} \quad (\sum_{k \in \Gamma_{\varepsilon}(\alpha)} |\alpha_k|+\varepsilon) (\sum_{l \in \Gamma_{\varepsilon}(\beta)} |\beta_l|+\varepsilon)
\end{equation}
$$
\mbox{subject to} \quad
|\alpha_k| > 0, |\beta_l| > 0, \sum_{k \in \Gamma_{\varepsilon}(\alpha)}  |\alpha_k|^2 \geq 1-\varepsilon^2, \sum_{l \in \Gamma_{\varepsilon} (\beta)}  |\beta_l|^2 \geq 1-\varepsilon^2.
$$
This can be separated into two optimization problems:
\begin{equation}\label{Eq7.2}
\max_{\alpha_k} \sum_{k \in  \Gamma_{\varepsilon}(\alpha)} |\alpha_k|
\quad \mbox{subject to} \quad
|\alpha_k| > 0, \sum_{k \in \Gamma_{\varepsilon}(\alpha)}  |\alpha_k|^2\geq1-\varepsilon^2
\end{equation}

\begin{equation}\label{Eq7.3}
\max_{\beta_l} \sum_{l \in \Gamma_{\varepsilon}(\beta)} |\beta_l|
\quad \mbox{subject to} \quad
|\beta_l| > 0, \sum_{l \in \Gamma_{\varepsilon}(\beta)}  |\beta_l|^2 \geq 1-\varepsilon^2.
\end{equation}

The optimization (\ref{Eq7.2}) can be solved by
\begin{equation}\label{Eq7.4}
\max_{\alpha_k} \sum_{k \in  \Gamma_{\varepsilon}(\alpha)} |\alpha_k|
\quad \mbox{subject to} \quad
|\alpha_k| > 0, \sum_{k \in \Gamma_{\varepsilon}(\alpha)}  |\alpha_k|^2=C,
\end{equation}
where $C \in [1-\varepsilon^2,1]$ is a constant.

By using Lagrangian multiplier method, we have Lagrangian function
$$
F(|\alpha_k|,\lambda)=\sum_{k \in  \Gamma_{\varepsilon}(\alpha)} |\alpha_k| + \lambda (\sum_{k \in \Gamma_{\varepsilon}(\alpha)} |\alpha_k|^2 - C).
$$

\noindent
Let the partial differatiation
$$
\frac{\partial F(|\alpha_k|,\lambda)}{\partial |\alpha_k|}=1+2\lambda |\alpha_k|=0.
$$
Then we have all the $|\alpha_k|$ are equal when $k \in \Gamma_{\varepsilon}(\alpha)$.
Denote $\| \alpha \|_{0(\varepsilon)}=A$.
Then $|\alpha_k|=\sqrt{\frac{C}{A}}$.
Hence the maxima of (\ref{Eq7.4}) is $A\sqrt{\frac{C}{A}}=\sqrt{AC}$.
Since $C \in [1-\varepsilon^2,1]$, then the maxima of (\ref{Eq7.2}) is $\sqrt{A}=\sqrt{\| \alpha \|_{0(\varepsilon)}}$.

Similarly, the maxima of (\ref{Eq7.3}) is $\sqrt{\| \beta \|_{0(\varepsilon)}}$.
Therefore, the maxima of (\ref{Eq7.1}) is
$(\sqrt{\| \alpha \|_{0(\varepsilon)}}+\varepsilon)\cdot (\sqrt{\| \beta \|_{0(\varepsilon)}}+\varepsilon)$,
and
$$
1
\leq
\mu (\sum_{k \in \Gamma_{\varepsilon}(\alpha)} |\alpha_k|+\varepsilon) (\sum_{l \in \Gamma_{\varepsilon}(\beta)} |\beta_l|+\varepsilon)
\leq
\mu (\sqrt{\| \alpha \|_{0(\varepsilon)}}+\varepsilon)\cdot (\sqrt{\| \beta \|_{0(\varepsilon)}}+\varepsilon).
$$
Using the inequality between the geometric and arithmetic means, we have
$$
1
\leq
\mu (\sqrt{\| \alpha \|_{0(\varepsilon)}}+\varepsilon)\cdot (\sqrt{\| \beta \|_{0(\varepsilon)}}+\varepsilon)
\leq
\mu \frac{{(\sqrt{\| \alpha \|_{0(\varepsilon)}}+\varepsilon)}^2+{(\sqrt{\| \beta \|_{0(\varepsilon)}}+ \varepsilon})^2}{2}.
$$
That is
$$ {(\sqrt{\|\alpha\|_{0(\varepsilon)}}+\varepsilon)}^2+{(\sqrt{\|\beta\|_{0(\varepsilon)}}+ \varepsilon)}^2 \geq \frac{2}{\mu}.$$
\hfill $\Box$

\noindent
\textit{\textbf{Proof of Theorem \ref{Unique}.}}
Suppose there are two different sparse representations of transfer function $H(z)$ under the two ORF bases
$\{\phi_k(z)\}_{k=1}^\infty$ and $\{\psi_l(z)\}_{l=1}^\infty$, that is
$$
H(z)=\sum_{k=1}^\infty \theta_k^{\phi} \phi_k(z)+\sum_{l=1}^\infty \theta_k^{\psi} \psi_l(z)
=\sum_{k=1}^\infty \xi_k^{\phi} \phi_k(z)+\sum_{l=1}^\infty \xi_l^{\psi} \psi_l(z)
$$
and
$${(\sqrt{\|\theta_1\|_{0(\varepsilon)} }+\varepsilon)}^2+{(\sqrt{\|\theta_2\|_{0(\varepsilon)} }+ \varepsilon)}^2 < \frac{1}{\mu},$$
$${(\sqrt{\|\xi_1\|_{0(\varepsilon)} }+\varepsilon)}^2+{(\sqrt{\|\xi_2\|_{0(\varepsilon)}}+ \varepsilon)}^2 < \frac{1}{\mu},$$
where
$\xi_1=[\xi_1^{\phi}, \xi_2^{\phi}, \cdots]^T$ and $\xi_2=[\xi_1^{\psi}, \xi_2^{\psi}, \cdots]^T.$ Then
$$\sum_{k=1}^\infty (\theta_k^{\phi} - \xi_k^{\phi}]) \phi_k(z)
=\sum_{l=1}^\infty (\xi_l^{\psi}-\theta_l^{\psi}) \psi_l(z).$$
According to the uncertainty principle, we have
\begin{equation} \label{Ine}
{(\sqrt{\|\theta_1-\xi_1\|_{0(\varepsilon)} }+\varepsilon)}^2+{(\sqrt{\|\theta_2-\xi_2\|_{0(\varepsilon)} }+ \varepsilon)}^2 \geq \frac{2}{\mu}.
\end{equation}
However, based on the sparsity assumption of Theorem \ref{UP}
\begin{eqnarray*}
&&{(\sqrt{\|\theta_1-\xi_1\|_{0(\varepsilon)} }+\varepsilon)}^2+{(\sqrt{\|\theta_2-\xi_2\|_{0(\varepsilon)} }+ \varepsilon)}^2\\
&<&{(\sqrt{(\|\theta_1\|_{0(\varepsilon)}+\|\xi_1\|_{0(\varepsilon)}) }+\varepsilon)}^2+{(\sqrt{(\|\theta_2\|_{0(\varepsilon)}+\|\xi_2\|_{0(\varepsilon)})}+ \varepsilon)}^2\\
&<& {(\sqrt{\|\theta_1\|_{0(\varepsilon)} }+\varepsilon)}^2
+{(\sqrt{\|\xi_1\|_{0(\varepsilon)}}+\varepsilon)}^2 + {(\sqrt{\|\theta_2\|_{0(\varepsilon)} }+ \varepsilon)}^2
+{(\sqrt{\|\xi_2\|_{0(\varepsilon)} }+ \varepsilon)}^2\\
&<& \frac{2}{\mu},
\end{eqnarray*}
which contradicts (\ref{Ine}).
\hfill $\Box$

\noindent
\textit{\textbf{Proof of Theorem \ref{ThOrho}.}}
The integral definition of inner product shows that
when $N$ is sufficiently large,
$$
\frac{1}{2\pi} \sum_{r=1}^N z_r^{-k} \overline{z_r^{-l}}\frac{2\pi}{N}
\rightarrow
\langle z^{-k},z^{-l}\rangle
=\delta_{kl},
$$
where the kronecker symbol $\delta_{kl}$ equals 1 if $k=l$ and 0 if $k \neq l$.

Then the $(k, l)$ element of $\Phi^{*} \Phi$ is
\begin{eqnarray*}
\sum_{r=1}^N \overline{\phi_k(z_r)} \phi_l(z_r)
&=&\sum_{r=1}^N \overline{\sum_{d'=0}^{\infty} b_{d'k} z_r^{-d'}} \sum_{d=0}^{\infty} b_{dl} z_r^{-d}\\
&=&\sum_{d'=0}^{\infty} \sum_{d=0}^{\infty} b_{d'k} b_{dl} \sum_{r=1}^N \overline{z_r^{-d'}} z_r^{-d}\\
&\rightarrow&
N\sum_{d=0}^{\infty} b_{dk} b_{dl}
=N \delta_{kl},
\end{eqnarray*}
the last equation is based on the orthonormality of $\{\psi_l(z)\}$,
\begin{eqnarray*}
\langle \phi_k(z), \phi_l(z)\rangle &=& \langle \sum_{d'=0}^{\infty}b_{d'k}z^{-d'}, \sum_{d=0}^{\infty}b_{dl}z^{-d} \rangle\\
\end{eqnarray*}
\begin{eqnarray*}
&=& \sum_{d'=d}^{\infty} \langle b_{d'k}z^{-d'},b_{d'l}z^{-d} \rangle + \sum_{d' \neq d} \langle b_{d'k}z^{-d'},b_{dl}z^{-d} \rangle\\
&=& \sum_{d=0}^{\infty} b_{dk} b_{dl} \langle z^{-d}, z^{-d} \rangle + \sum_{d'\neq d} b_{d'k} b_{dl} \langle z^{-d'},z^{-d} \rangle\\
&=& \sum_{d=0}^{\infty} b_{dk} b_{dl}=\delta_{kl},\\
\end{eqnarray*}
which implies (1).

The $(k, l)$ element of $\Psi^{*} \Psi$ is
\begin{eqnarray*}
\sum_{r=1}^N \overline{\psi_k(z_r)} \psi_l(z_r)
&=&\sum_{r=1}^N \overline{\sum_{d'=0}^{\infty} a_{d'k} z_r^{-d'}} \sum_{d=0}^{\infty} a_{dl} z_r^{-d}\\
&=&\sum_{d'=0}^{\infty} \sum_{d=0}^{\infty} a_{d'k} a_{dl} \sum_{r=1}^N \overline{z_r^{-d'}} z_r^{-d}\\
 &\rightarrow &
N\sum_{d=0}^{\infty} a_{dk} a_{dl}
=N \delta_{kl},
\end{eqnarray*}
the last equation is based on the orthonormality of $\{\psi_l(z)\}$,
\begin{eqnarray*}
\langle \psi_k(z), \psi_l(z)\rangle &=& \langle \sum_{d'=0}^{\infty}a_{d'k}z^{-d'}, \sum_{d=0}^{\infty}a_{dl}z^{-d} \rangle\\
&=& \sum_{d'=d}^{\infty} \langle a_{d'k}z^{-d'},a_{d'l}z^{-d} \rangle + \sum_{d' \neq d} \langle a_{d'k}z^{-d'},a_{dl}z^{-d} \rangle\\
&=& \sum_{d=0}^{\infty} a_{dk} a_{dl} \langle z^{-d}, z^{-d} \rangle + \sum_{d'\neq d} a_{d'k} a_{dl} \langle z^{-d'},z^{-d} \rangle\\
&=&\sum_{d=0}^{\infty} a_{dk} a_{dl}=\delta_{kl},\\
\end{eqnarray*}
which implies (2).

The $(k, l)$ element of $\Phi^{*} \Psi$ is
\begin{eqnarray*}
\sum_{r=1}^N \overline{\phi_l(z_r)} \psi_l(z_r)
&=&\sum_{r=1}^N \overline{\sum_{d'=0}^{\infty}b_{d'k} z_r^{-d'}}\sum_{d=0}^{\infty}a_{dl} z_r^{-d}\\
&=&\sum_{d'=0}^{\infty} \sum_{d=0}^{\infty} b_{d'k} a_{dl} \sum_{r=1}^N \overline{z_r^{-d'}} z_r^{-d}\\
&\rightarrow&
N \sum_{d=0}^{\infty} b_{dk} a_{dl}.
\end{eqnarray*}
\hfill $\Box$

\section{Proof of Theorem  \ref{Th7}}
{\hspace{5mm}
Now, we will give the general idea for proving Theorem \ref{Th7}.

The matrix $[\Phi_{\Omega T_1} \quad \Psi_{\Omega T_2}]^* [\Phi_{\Omega T_1} \quad \Psi_{\Omega T_2}]$ can be written as
$$
[\Phi_{\Omega T_1} \quad \Psi_{\Omega T_2}]^* [\Phi_{\Omega T_1} \quad \Psi_{\Omega T_2}]
=\sum_{k=1}^N \delta_k [u^k \quad v^k] \otimes [u^k \quad v^k]\\
=\sum_{k=1}^N \delta_k
\begin{bmatrix}
u^k \otimes u^k & u^k \otimes v^k\\
v^k \otimes u^k & v^k \otimes v^k
\end{bmatrix},
$$
where $u^k$ and $v^k$ are the row vectors of $\Phi_{T_1}$ and $\Psi_{T_2}$, respectively,
and $\delta_k$ obeys (\ref{delta}).

Denote
\begin{equation} \label{Y}
Y
=\frac{1}{m}
[\Phi_{\Omega T_1} \quad \Psi_{\Omega T_2}]^* [\Phi_{\Omega T_1} \quad \Psi_{\Omega T_2}] - F
=\frac{1}{m}
\sum_{k=1}^N \delta_k
\begin{bmatrix}
u^k \otimes u^k & u^k \otimes v^k\\
v^k \otimes u^k & v^k \otimes v^k
\end{bmatrix}
- F.
\end{equation}

The proof includes two steps:
(1) $\mathbf{E}\|Y\|$ is upper bounded with small constant,
i.e. on average that matrix $\frac{1}{m}
[\Phi_{\Omega T_1} \quad \Psi_{\Omega T_2}]^* [\Phi_{\Omega T_1} \quad \Psi_{\Omega T_2}]$ deviated little from $F$.
(2) $\|Y\|-\mathbf{E}\|Y\|$ is bounded with high probability.
We present these results in Lemma 3 and Lemma 4, respectively.

\begin{lemma} \label{Lemma5}
Let $[\Phi \quad \Psi]$, $T$ and $\Omega$ be as defined in Theorem \ref{Th7}.
Then
$$
\mathbf{E}\|\frac{1}{m}[\Phi_{\Omega T_1} \quad \Psi_{\Omega T_2}]^* [\Phi_{\Omega T_1} \quad \Psi_{\Omega T_2}]-F\|
\leq
C_R \frac{\sqrt{1+3\|F\|}}{2} \cdot \mu_M \frac{\sqrt{|T|\log{|T|}}}{\sqrt{m}},
$$
for some positive constant $C_R$,
provided that $C_R/2\cdot \frac{\sqrt{\log{|T|}}}{\sqrt{m}} \cdot \max_{1 \leq k \leq N} \|[u^k \quad v^k]\|$ is less than 1.
\end{lemma}

\begin{proof}
First we have
\begin{eqnarray*}
\mathbf{E}Y&=&
\begin{bmatrix}
\mathbf{E}\frac{1}{m}\sum_{k=1}^N \delta_k u^k \otimes u^k-I & E\frac{1}{m}\sum_{k=1}^N \delta_k u^k \otimes v^k- \mu \langle \Phi_{T_1}, \Psi_{T_2} \rangle\\
\mathbf{E}\frac{1}{m} \sum_{k=1}^N \delta_k v^k \otimes u^k - \mu \langle \Phi_{T_1}, \Psi_{T_2} \rangle   & \mathbf{E}\frac{1}{m} \sum_{k=1}^N \delta_k v^k \otimes v^k-I
\end{bmatrix}\\
&=& \begin{bmatrix}
\frac{1}{m}\sum_{k=1}^N \frac{m}{N} u^k \otimes u^k-I & \frac{1}{m}\sum_{k=1}^N \frac{m}{N} u^k \otimes v^k- \mu \langle \Phi_{T_1}, \Psi_{T_2} \rangle\\
\frac{1}{m}\sum_{k=1}^N \frac{m}{N} v^k \otimes u^k - \mu \langle \Phi_{T_1}, \Psi_{T_2} \rangle   & \frac{1}{m}\sum_{k=1}^N \frac{m}{N} v^k \otimes v^k-I
\end{bmatrix}\\
&=& \begin{bmatrix}
\frac{1}{N}\sum_{k=1}^N  u^k \otimes u^k-I & \frac{1}{N}\sum_{k=1}^N  u^k \otimes v^k- \mu \langle \Phi_{T_1}, \Psi_{T_2} \rangle\\
\frac{1}{N}\sum_{k=1}^N  v^k \otimes u^k - \mu \langle \Phi_{T_1}, \Psi_{T_2} \rangle   & \frac{1}{N}\sum_{k=1}^N v^k \otimes v^k-I
\end{bmatrix}\\
&=& O.
\end{eqnarray*}

Use the symmetrization technique in [44],
and let $\delta'_1, \cdots, \delta'_n$ are independent copies of $\delta_1, \cdots, \delta_n$,
and  $\epsilon_1, \cdots, \epsilon_n$ be a sequence of Bernoulli variables taking values $\pm1$ with probability $\frac{1}{2}$ (and independent of sequences $\delta$ and $\delta'$),
we have
\begin{eqnarray*}
\mathbf{E}\|Y\|
&\leq& \mathbf{E}_{\delta, \delta'} \|\frac{1}{m} \sum_{k=1}^N (\delta_k-\delta'_k) [u^k \quad v^k]  \otimes [u^k \quad v^k] \|\\
&=& E_{\epsilon} \mathbf{E}_{\delta, \delta'} \|\frac{1}{m} \sum_{k=1}^N \epsilon_k (\delta_k-\delta'_k) [u^k \quad v^k]  \otimes [u^k \quad v^k] \|\\
& \leq & 2 \mathbf{E}_{\epsilon} \mathbf{E}_{\delta} \|\frac{1}{m} \sum_{k=1}^N \epsilon_k  \delta_k [u^k \quad v^k]  \otimes [u^k \quad v^k] \|,
\end{eqnarray*}
the first equality follows from the symmetry of the random variable
$(\delta_k-\delta'_k) [u^k \quad v^k] \otimes [u^k \quad v^k]$
while the last inequality follows from the triangle inequality.

Rudelson's lemma [50] states that
\[
\mathbf{E}_{\epsilon} \|\sum_{k=1}^N \epsilon_k  \delta_k [u^k \quad v^k]  \otimes [u^k \quad v^k] \|
\leq \frac{C_R}{4}\cdot \sqrt{\log{|T|}} \cdot \max_{k: \delta_k=1} \|[u^k \quad v^k] \|
\cdot \sqrt{\|\sum_{k=1}^N  \delta_k [u^k \quad v^k] \otimes [u^k \quad v^k] \|},
\]
for some universal constant $C_R>0$.

Taking expectation over $\delta$ then gives
\begin{eqnarray*}
 \mathbf{E}\|Y\|
& \leq &  C_R/2\cdot \frac{\sqrt{\log{|T|}}}{m} \cdot \max_{1 \leq k \leq N} \|[u^k \quad v^k] \| \cdot \mathbf{E}\sqrt{\|\sum_{k=1}^N  \delta_k [u^k \quad v^k]  \otimes [u^k \quad v^k] \|}\\
& \leq &  C_R/2\cdot \frac{\sqrt{\log{|T|}}}{m} \cdot \max_{1 \leq k \leq N} \|[u^k \quad v^k] \| \cdot \sqrt{\mathbf{E}\|\sum_{k=1}^N  \delta_k [u^k \quad v^k]  \otimes [u^k \quad v^k] \|}.\\
\end{eqnarray*}

Observe that
$$
\mathbf{E}\|\sum_{k=1}^N  \delta_k [u^k \quad v^k] \otimes [u^k \quad v^k] \|
=\mathbf{E}\|mY+mF\|
\leq m(\mathbf{E}\|Y\|+\|F\|).
$$
Therefore
$$
\mathbf{E} \|Y\| \leq a \sqrt{\mathbf{E}\|Y\|+\|F\|},
$$
where $a=C_R/2\cdot \frac{\sqrt{\log{|T|}}}{\sqrt{m}} \cdot \max_{1 \leq k \leq N} \|[u^k \quad v^k]\|$,
which implies
$$
0 \leq \mathbf{E}\|Y\| \leq \frac{a^2+\sqrt{a^4+4a^2\|F\|}}{2}.
$$
It then follows that if $a \leq 1$,
$$\mathbf{E}\|Y\| \leq \sqrt{1+3\|F\|} a.$$
The inequality is based on the fact: when $a \leq 1$,
\begin{eqnarray*}
(a^2+\sqrt{a^4+4a^2\|F\|})^2
&=& 2a^4+2a^3\sqrt{a^2+4\|F\|}+4a^2\|F\|\\
&\leq & 2a^2+2a^2\sqrt{1+4\|F\|}+4a^2\|F\|\\
&\leq & (4+12\|F\|)a^2.
\end{eqnarray*}
Since
$$
\max_{1 \leq k \leq N} \|[u^k \quad v^k]\| \leq \sqrt{|T_1| \mu_{\Phi}^2+|T_2| \mu_{\Psi}^2}
\leq \sqrt{|T_1| +|T_2|} \max \{\mu_{\Phi}, \mu_{\Psi}\}
=\sqrt{|T|} \mu_M,
$$
then
$$
a \leq C_R/2\cdot \frac{\sqrt{\log{|T|}}}{\sqrt{m}} \cdot \sqrt{|T|} \mu_M,
$$
which completes the proof.
\end{proof}

\begin{lemma} \label{l}
Let $[\Phi \quad \Psi]$, $T$ and $\Omega$ be as defined in Theorem \ref{Th7}.
Then for all $t \geq 0$,
$$
P\{|\|Y\|-\mathbf{E}\|Y\||>t\} \leq 3 \exp \left [-\frac{tm}{K|T|\mu_M^2} \log(1+\frac{t}{2+ \mathbf{E}\|Y\|})  \right ].
$$
where $K$ is a numerical constant.
\end{lemma}

The proof of Lemma \ref{l} uses a concentration equality about the large deviation of suprema of sums of
independent variables [47], which is stated below.

\begin{lemma} \label{Taka}
Let $Y_1, Y_2, \cdots, Y_n$ be a sequence of independent random variables taking values in a Banach space
and let $Z$ be the supremum defined as
$$Z=\sup_{f \in \mathcal{F}}\sum_{i=1}^n f(Y_i),$$
where $\mathcal{F}$ is a countable family of real-valued functions.
Assume that $|f| \leq B$ for every $f$ in $\mathcal{F}$,
and  $\mathbf{E}f(Y_i)=0$ for every $f$ in $\mathcal{F}$ and $i= 1,2,\cdots, n$.
Then for all $t \geq 0$,
$$
P(|Z-\mathbf{E}Z|>t) \leq 3 \exp \left [-\frac{t}{KB} \log(1+\frac{Bt}{\sigma^2+B \mathbf{E}\bar{Z}})  \right ],
$$
where $\sigma^2=\sup_{f \in \mathcal{F}} \sum_{k=1}^n \mathbf{E} f^2(Y_k)$,
$\bar{Z}=\sup_{f \in \mathcal{F}}|\sum_{i=1}^n f(Y_i)|$,
and $K$ is a numerical constant.
\end{lemma}

\begin{proof}
Since
$
\frac{1}{N}\sum_{k=1}^N  [u^k \quad v^k] \otimes [u^k \quad v^k]= F
$,
then express $Y$ in (\ref{Y}) as
$$
Y=\sum_{k=1}^N (\delta_k- \frac{m}{N}) \frac{[u^k \quad v^k] \otimes [u^k \quad v^k]} {m}:=\sum_{k=1}^N Y_k,
$$
where
$Y_k:=(\delta_k- \frac{m}{N}) \frac{[u^k \quad v^k] \otimes [u^k \quad v^k]} {m}$.

According the definition of spectral norm,
$$
\|Y\|=\sup_{f_1,f_2} \langle f_1, Y f_2 \rangle
=\sup_{f_1,f_2} \sum_{k=1}^N \langle f_1, Y_k f_2 \rangle,
$$
where the supremum is over a countable collection of unit vectors.

Let $f(Y_k)$ denote the mapping $\langle f_1, Y_k f_2 \rangle$,
then we have $\mathbf{E}f(Y_k)=0$ and for all $k$
$$
|f(Y_k)|=|\langle f_1, Y_k f_2 \rangle|
\leq \frac{|\langle f_1, [u^k \quad v^k] \rangle \langle [u^k \quad v^k] ,f_2\rangle|} {m}
\leq \frac{\|[u^k \quad v^k]\|^2} {m}
\leq B,
$$
where $B=\max_{1 \leq k \leq N}  \frac{\|[u^k \quad v^k]\|^2} {m}$.

Hence according to Lemma \ref{Taka},
for all $t \geq 0$,

\begin{equation} \label{Yb}
P(|\|Y\|-\mathbf{E}\|Y\||>t ) \leq 3 \exp \left [-\frac{t}{KB} \log(1+\frac{Bt}{\sigma^2+B \mathbf{E}\|Y\|})  \right ],
\end{equation}
where $\sigma^2=\sup_f \sum_{k=1}^N \mathbf{E} f^2(Y_k)$.

We now compute
\begin{eqnarray*}
\mathbf{E}f^2(Y_k)
&=& \frac{m}{N} (1- \frac{m}{N}) \frac{|\langle f_1, [u^k \quad v^k] \rangle \langle [u^k \quad v^k] ,f_2\rangle|^2} {m^2}\\
&\leq& \frac{m}{N} (1- \frac{m}{N}) \frac{\|[u^k \quad v^k]\|^2} {m^2} |\langle [u^k \quad v^k] ,f_2\rangle|^2,
\end{eqnarray*}
the inequality is based on Cauchy inequality, and

\begin{eqnarray*}
\sum_{k=1}^N \mathbf{E}f^2(Y_k)
&\leq& \sum_{k=1}^N \frac{m}{N} (1- \frac{m}{N}) \frac{\|[u^k \quad v^k]\|^2} {m^2} |\langle [u^k \quad v^k] ,f_2\rangle|^2\\
&\leq&  \frac{1}{N} (1- \frac{m}{N}) \max_{1\leq k\leq N}\frac{\|[u^k \quad v^k]\|^2} {m} \sum_{k=1}^N|\langle [u^k \quad v^k] ,f_2\rangle|^2\\
&\leq&  \frac{1}{N} (1- \frac{m}{N}) B \sum_{k=1}^N|\langle [u^k \quad v^k] ,f_2\rangle|^2\\
&\leq & 2B,
\end{eqnarray*}
the last inequality is based on
\begin{eqnarray}
&& \sum_{k=1}^N|\langle [u^k \quad v^k] ,f_2\rangle|^2 \nonumber\\
&=& \sum_{k=1}^N (\Phi_{k1}f_{2, 1} +\cdots + \Phi_{kT_1}f_{2, T_1} + \Psi_{k1}f_{2, T_1+1} +\cdots + \Phi_{kT_2}f_{2, T_1+T_2})^2 \nonumber\\
&\leq & 2 \sum_{k=1}^N (\Phi_{k1}f_{2, 1} +\cdots + \Phi_{kT_1}f_{2, T_1})^2+ 2 \sum_{k=1}^N (\Psi_{k1}f_{2, T_1+1} +\cdots + \Phi_{kT_2}f_{2, T_1+T_2})^2  \nonumber \\
&= & 2 \sum_{k=1}^N [(\Phi_{k1}f_{2, 1})^2 +\cdots + (\Phi_{kT_1}f_{2, T_1})^2]+ 2 \sum_{k=1}^N [(\Psi_{k1}f_{2, T_1+1})^2 +\cdots + (\Psi_{kT_2}f_{2, T_1+T_2})^2] \nonumber \\
&=& 2 N(f_{2, 1}^2 +\cdots + f_{2, T_1}^2)+ 2 N (f_{2, T_1+1}^2 +\cdots + f_{2, T_1+T_2}^2) \nonumber \\
&\leq & 2N, \label{2N}
\end{eqnarray}
where $f_{2, j}$ is the $j$-th element of $f_2$.
And the second and third equalities are based on (\ref{ortho}).

Then $\sigma^2=2B$. So equation (\ref{Yb}) can be rewritten as
$$
P(|\|Y\|-\mathbf{E}\|Y\||>t) \leq 3 \exp \left [-\frac{t}{KB} \log(1+\frac{t}{2+ \mathbf{E}\|Y\|})  \right ].
$$

Since $B \leq \frac{|T|\mu_M^2}{m}$,
which implies
$$
P(|\|Y\|-\mathbf{E}\|Y\||>t) \leq 3 \exp \left [-\frac{tm}{K|T|\mu_M^2} \log(1+\frac{t}{2+ \mathbf{E}\|Y\|})  \right ].
$$
\end{proof}

\noindent
\textit{\textbf{Proof of Theorem \ref{Th7}.}}
By Lemma \ref{Lemma5},
taking $m \geq 4 C_R^2 (1+3\|F\|)\mu_M^2 |T|\log|T|$ so that $\mathbf{E}\|Y\| \leq \frac{1}{4}$.
And pick $t=\frac{1}{4}$,
then by Lemma 4
$$
P(\|Y\|> 1/2) \leq 3 \exp{\left (-\frac{m}{C_T |T| \mu_M^2}\right ) },
$$
where $C_T=4K/\log{(10/9)}$.

Taking $m\geq C_T |T| \mu_M^2 \log {(\frac{3}{\delta})}$ so that
$\exp{\left (-\frac{m}{C_T |T| \mu_M^2}\right ) }<\delta$,
which complete the proof of the first result (1) in Theorem \ref{Th7}.

For the result (2), we have
\begin{eqnarray*}
&&P(\|\frac{1}{m} [\Phi_{\Omega T_1} \quad \Psi_{\Omega T_2}]^* [\Phi_{\Omega T_1} \quad \Psi_{\Omega T_2}]
- I\|> 1/2 + \|F-I\| ) \\
&=& P(\|\frac{1}{m} [\Phi_{\Omega T_1} \quad \Psi_{\Omega T_2}]^* [\Phi_{\Omega T_1} \quad \Psi_{\Omega T_2}]
- F +F -I\|> 1/2 + \|F-I\| ) \\
& \leq & P(\|\frac{1}{m} [\Phi_{\Omega T_1} \quad \Psi_{\Omega T_2}]^* [\Phi_{\Omega T_1} \quad \Psi_{\Omega T_2}]
- F\| +\|F -I\|> 1/2 + \|F-I\| ) \\
&= & P(\|\frac{1}{m} [\Phi_{\Omega T_1} \quad \Psi_{\Omega T_2}]^* [\Phi_{\Omega T_1} \quad \Psi_{\Omega T_2}]
- F\| > 1/2  ) \\
&\leq & \delta.
\end{eqnarray*}

By simple calculation, we have
$$
\|F-I\|=\|\Phi_{T_1}^*\Psi_{T_2} /N\|,
$$
which completes the proof.
$\hfill{} \Box$

\section{Proof of Theorem \ref{Th8}}
{\hspace{5mm}
In the sequel, we will show that  $\pi(t) <1$ for $t \in T^c$ with high probability.

For a particular $t_0 \in T^c$, rewrite
$$
\pi(t_0)
=\langle v^0,[[\Phi_{\Omega T_1} \quad \Psi_{\Omega T_2}]^* [\Phi_{\Omega T_1} \quad \Psi_{\Omega T_2}]]^{-1} \tau \rangle
=\langle w^0, z\rangle,
$$
where $v^0$ is the $t_0$-th row vector of $[\Phi_\Omega \quad \Psi_\Omega]^*[\Phi_{\Omega T_1} \quad \Psi_{\Omega T_2}]$ and
$w^0=[[\Phi_{\Omega T_1} \quad \Psi_{\Omega T_2}]^*$ $[\Phi_{\Omega T_1}$ \quad $\Psi_{\Omega T_2}]]^{-1} v^0$.

Set $\lambda_k^0=[\Phi \quad \Psi]_{k t_0}$.
The vector $v^0$ is given by
$$
v^0=\sum_{k=1}^N \delta_k \lambda_k^0 [u^k \quad v^k].
$$

Let
$$
\tilde{v}^0=v^0-\mathbf{E}v^0=\sum_{k=1}^N (\delta_k- \mathbf{E} \delta_k) \lambda_k^0 [u^k \quad v^k].
$$

The following lemmas give estimates for the size of these vectors.

\begin{lemma} \label{Lemma8}
The second moment of $\|\tilde{v}^0\|$ and $\|v^0\|$ obeys
$$
\mathbf{E} \|\tilde{v}^0\|^2
\leq  m |T|\mu_M^2
$$
and
$$
\mathbf{E} \|v^0\|^2
\leq 2m |T| \mu_M^2 + 2 m^2 \mu^2 \max\{|T_1|,|T_2|\},
$$
respectively.
\end{lemma}

\begin{proof}
$\tilde{v}^0$ can be viewed as a sum of independent random variables:
$$
\tilde{v}^0=\sum_{k=1}^N Y_k, \quad Y_k=(\delta_k- \mathbf{E} \delta_k) \lambda_k^0 [u^k \quad v^k],
$$
where $\mathbf{E} Y_k=0$.
It follows that
\begin{eqnarray*}
\mathbf{E} \|\tilde{v}^0\|^2
&=&\sum_{k}\mathbf{E}\langle Y_k,Y_k \rangle + \sum_{k'\neq k}\mathbf{E}\langle Y_k,Y_k' \rangle \\
&=&\sum_{k}\mathbf{E}\langle Y_k,Y_k \rangle\\
&=&\sum_{k} \frac{m} {N} (1- \frac{m} {N})|\lambda_k^0|^2  \|[u^k \quad v^k]\|^2\\
&\leq&\frac{m} {N} (1- \frac{m} {N}) (|T_1|\mu_\Phi^2+|T_2|\mu_\Psi^2)\sum_{k}|\lambda_k^0|^2\\
&=& m (1- \frac{m} {N}) (|T_1|\mu_\Phi^2+|T_2|\mu_\Psi^2)\\
&\leq& m |T|\mu_M^2,
\end{eqnarray*}
the last equality is based on $\sum_{k}|\lambda_k^0|^2=N$.
Hence
\begin{eqnarray*}
\mathbf{E} \|v^0\|^2
&=& \mathbf{E} \| \tilde{v}^0 +\sum_{k=1}^N \mathbf{E} \delta \lambda_k^0 [u^k \quad v^k]\|^2\\
&=& \mathbf{E} \| \tilde{v}^0 + \frac{m} {N} \sum_{k=1}^N \lambda_k^0 [u^k \quad v^k]\|^2\\
&\leq& \mathbf{E} 2 (\| \tilde{v}^0 \|^2+ \|\frac{m} {N} \sum_{k=1}^N \lambda_k^0 [u^k \quad v^k]\|^2)\\
\end{eqnarray*}

Notice that
if $t_0 \in T_1^c$, then
$$
\sum_{k=1}^N \lambda_k^0 [u^k \quad v^k]
=\Phi_{t_0}^*[\Phi_{T_1} \quad \Psi_{T_2}]
=[\overrightarrow{0} \quad \sum_{k=1}^N \lambda_k^0 v^k],
$$
which implies that
$$
\| \frac{m}{N}\sum_{k=1}^N \lambda_k^0 [u^k \quad v^k]\|
=\|\frac{m}{N} \sum_{k=1}^N \lambda_k^0 v^k\|
=m \| \frac{\langle\Phi_{t_0}, \Psi_{T_2} \rangle}{N}\|
\leq m \mu \sqrt{|T_2|}.
$$
Similarly, if $t_0 \in T_2^c$, then
$$
\| \frac{m}{N} \sum_{k=1}^N \lambda_k^0 [u^k \quad v^k]\|
\leq  m \mu \sqrt{|T_1|}.
$$

In summary,
if $t_0 \in T^c$, then
\begin{equation} \label{epc}
\| \frac{m}{N} \sum_{k=1}^N \lambda_k^0 [u^k \quad v^k]\|
\leq  m \mu \max\{\sqrt{|T_1|},\sqrt{|T_2|}\}.
\end{equation}
\noindent
Hence,
$$
\mathbf{E} \|v^0\|^2
\leq 2m |T| \mu_M^2 + 2 m^2 \mu^2 \max\{|T_1|,|T_2|\}.
$$
\end{proof}

\begin{lemma} \label{Lemma9}
Fix $t_0 \in T^c$.
Define $\bar{\sigma}$ as
$$
\bar{\sigma}^2=\max\{ 2m \mu_M^2 , \sqrt{m}|T| \mu_M^3\}= m \mu_M^2 \max\{2 , \frac{|T| \mu_M}{\sqrt{m}}\}.
$$
For $a>0$ obeying $a \leq \frac{ \sqrt{2 m}}{\sqrt{|T|}\mu_M}$ if $\frac{|T| \mu_M}{\sqrt{m}} \leq 2$
and $a \leq ({\frac{m}{\mu_M^2}})^{1/4}$ otherwise.
Then
$$
P(\|v^0\| >a\bar{\sigma}+ m \mu \sqrt{|T|} + \sqrt{m |T|} \mu_M) \leq 3 e^{-\gamma a^2},
$$
for some positive constant $\gamma >0$.
\end{lemma}

\begin{proof}
By definition, $\|\tilde{v}^0\|$ is given by
$$
\|\tilde{v}^0\|
=\sup_{\|f\|=1}\langle \tilde{v}^0, f \rangle
=\sup_{\|f\|=1}\sum_{k=1}^N \langle Y_k, f \rangle.
$$

For a fixed unit vector $f$, let $f(Y_k)$ denote the mapping $\langle Y_k, f \rangle$.
Since $\mathbf{E} Y_k=0$, then $\mathbf{E} f(Y_k)=0$. And
\begin{eqnarray*}
|f(Y_k)|
&\leq& |\lambda_k^0||\langle f, [u^k \quad v^k]\rangle|\\
&\leq& |\lambda_k^0|\|[u^k \quad v^k]\|\\
&\leq& \max\{\mu_{\Phi}, \mu_{\Psi}\}\cdot \sqrt{|T_1| \mu_{\Phi}^2+|T_2| \mu_{\Psi}^2}\\
&\leq& \sqrt{|T|} \mu_M^2=B',
\end{eqnarray*}
for all $k$.

By applying (\ref{2N}), we have
$$
\sum_{k=1}^N \mathbf{E} f^2(Y_k)
=\sum_{k=1}^N\frac{m}{N}(1-\frac{m}{N})|\lambda_k^0|^2|\langle[u^k \quad v^k], f \rangle|^2
\leq \frac{m}{N}(1-\frac{m}{N})\mu_M^2 \cdot 2N
\leq 2m \mu_M^2,
$$
which implies $\sigma^2=2m \mu_M^2$.

According to Lemma \ref{Taka}, we have
$$
P(|\|\tilde{v}^0\|-\mathbf{E}\|\tilde{v}^0\|| >t)
\leq 3 \exp \left [-\frac{t}{KB'} \log(1+\frac{B't}{\sigma^2+B' \mathbf{E}\|\tilde{v}^0\|})  \right ],
$$
which implies
$$
P(\|\tilde{v}^0\| >t + \mathbf{E}\|\tilde{v}^0\|)
\leq 3 \exp \left [-\frac{t}{KB'} \log(1+\frac{B't}{\sigma^2+B' \mathbf{E}\|\tilde{v}^0\|})  \right ].
$$

For the $\mathbf{E}\|\tilde{v}^0\|$, we simply use
$$
\mathbf{E}\|\tilde{v}^0\|
\leq \sqrt{\mathbf{E}\|\tilde{v}^0\|^2}
\leq \mu_M \sqrt{m |T|}.
$$

Since
$
v^0=\tilde{v}^0 +\sum_{k=1}^N \mathbf{E} \delta \lambda_k^0 [u^k \quad v^k]
=\tilde{v}^0 + \frac{m}{N} \sum_{k=1}^N \lambda_k^0 [u^k \quad v^k],
$
then
\begin{eqnarray*}
&&P(\|v^0\| >t+ \frac{m}{N} \| \sum_{k=1}^N \lambda_k^0 [u^k \quad v^k]\| +\mathbf{E}\|\tilde{v}^0\| )\\
&=&P(\|\tilde{v}^0 +\frac{m}{N}\sum_{k=1}^N \lambda_k^0 [u^k \quad v^k]\|> t+ \frac{m}{N} \|\sum_{k=1}^N \lambda_k^0 [u^k \quad v^k]\| +\mathbf{E}\|\tilde{v}^0\|)\\
&\leq& P(\|\tilde{v}^0\| +\frac{m}{N}\|\sum_{k=1}^N \lambda_k^0 [u^k \quad v^k]\|>t+ \frac{m}{N} \|\sum_{k=1}^N \lambda_k^0 [u^k \quad v^k]\| +\mathbf{E}\|\tilde{v}^0\|)\\
&=&P(\|\tilde{v}^0\|>t +\mathbf{E}\|\tilde{v}^0\|)\\
& \leq &  3 \exp \left [-\frac{t}{KB'} \log(1+\frac{B't}{\sigma^2+B' \mathbf{E}\|\tilde{v}^0\|})  \right ]\\
& \leq &  3 \exp \left [-\frac{t}{KB'} \log \left (1+\frac{B't}{2m \mu_M^2 + \sqrt{m}|T| \mu_M^3}\right)  \right ].\\
\end{eqnarray*}

Suppose that $\bar{\sigma}^2=\max\{2m \mu_M^2 , \sqrt{m}|T| \mu_M^3\}= m \mu_M^2 \max\{ 2 , \frac{|T| \mu_M}{\sqrt{m}}\}$.
And fix $t=a \bar{\sigma}$,
then by using Tailor's expansion of logarithm function, we have
$$
P(\|v^0\| >t+ \frac{m}{N} \| \sum_{k=1}^N \lambda_k^0 [u^k \quad v^k]\| +\mathbf{E}\|\tilde{v}^0\| ) \leq 3 e^{-\gamma a^2},
$$
provided that $B't \leq \bar{\sigma}^2$,
which is equivalent to the following two situations:

(1) when $\frac{|T| \mu_M}{\sqrt{m}} \leq 2$, then $a \leq \frac{\sqrt{2 m}}{\sqrt{|T|}\mu_M}$.

(2) when $\frac{|T| \mu_M}{\sqrt{m}}>2$, then $a \leq ({\frac{m}{\mu_M^2}})^{1/4}$.

As seen (\ref{epc}) in Lemma \ref{Lemma8},
\begin{eqnarray*}
\frac{m}{N} \| \sum_{k=1}^N \lambda_k^0 [u^k \quad v^k]\|
&\leq&  m \mu (\Phi, \Psi) \max\{\sqrt{|T_1|},\sqrt{|T_2|}\}\\
&\leq&  m \mu (\Phi, \Psi) \sqrt{|T_1|+|T_2|}
=m \mu (\Phi, \Psi) \sqrt{|T|},
\end{eqnarray*}
then
$$
P(\|v^0\| >a\bar{\sigma}+ m \mu (\Phi, \Psi) \sqrt{|T|} + \sqrt{m |T|} \mu_M) \leq 3 e^{-\gamma a^2}.
$$
\end{proof}

\begin{lemma}\label{Lemma10}
Let $w^0=([\Phi_{\Omega T_1} \quad \Psi_{\Omega T_2}]^* [\Phi_{\Omega T_1} \quad \Psi_{\Omega T_2}])^{-1} v^0$.
With the same notations and hypotheses as in Lemma \ref{Lemma9}, then
\begin{eqnarray*}
&&P(\sup_{t_0 \in T^c} \|\omega^0\| \geq (\frac{1}{2}+\|\Phi_{T_1}^*\Psi_{T_2} /N\|)^{-1} (\mu (\Phi, \Psi) \sqrt{|T|}+ \frac{a\bar{\sigma}}{m}+ +\frac{ \sqrt{|T|} \mu_M}{\sqrt{m}})\\
&\leq& 3(n_1+n_2) e^{-\gamma a^2}
+P(\|[\Phi_{\Omega T_1} \quad \Psi_{\Omega T_2}]^* [\Phi_{\Omega T_1} \quad \Psi_{\Omega T_2}]\|
\leq (\frac{1}{2}+\|F-I\|) m )
\end{eqnarray*}
\end{lemma}

\begin{proof}

Let $A$ and $B$ be the events
$
\{\|[\Phi_{\Omega T_1} \quad \Psi_{\Omega T_2}]^* [\Phi_{\Omega T_1} \quad \Psi_{\Omega T_2}]\|
\geq (\frac{1}{2}+\|\Phi_{T_1}^*\Psi_{T_2} /N\|) m \}
$
and
$\{\sup_{t_0 \in T^c} \|v^0\| \leq a\bar{\sigma}+ m \mu (\Phi, \Psi) \sqrt{|T|} + \sqrt{m |T|} \mu_M\}$ respectively.

Lemma \ref{Lemma9} gives $P(B^c) \leq 3(n_1+n_2) e^{-\gamma a^2}$.
Then on the event
$
A\cap B =
\{\sup_{t_0 \in T^c} \|\omega^0\| \leq (\frac{1}{2}+\|\Phi_{T_1}^*\Psi_{T_2} /N\|)^{-1} (\mu (\Phi, \Psi) \sqrt{|T|}+ \frac{a\bar{\sigma}}{m}+ +\frac{ \sqrt{|T|} \mu_M}{\sqrt{m}})
$,
we have
$$
P({(A \cap B)}^c) =  P( A^c \cup B^c) \leq P( A^c )+ P( B^c ),
$$
the claim follows.
\end{proof}

\begin{lemma}\label{Lemma11}
Assume that $\tau(t), t \in T$ is an i.i.d. sequence of symmetric Bernoulli random variables.
For each $\lambda >0$, we have
$$
P(\sup_{t \in T^c} |\pi(t)| \geq 1 ) \leq 2 (n_1+n_2) e^{-\frac{1}{2 \lambda^2}}+ P(\sup_{t_0 \in T^c} \|\omega^0\| \geq \lambda)
$$
\end{lemma}

\begin{proof}
Recall that $\pi(t_0)=\langle \omega^0, \tau \rangle$,
by using Hoeffding's inequality, we have
$$
P(|\pi(t_0)|>1| \omega^0)
=P(|\langle \omega^0, \tau \rangle|>1| \omega^0)
\leq 2 e^{-\frac{1}{2\|\omega^0\|^2}}.
$$

It then follows that
$$
P(\sup_{t_0 \in T^c}|\pi(t_0)|>1| \sup_{t_0 \in T^c}\|\omega^0\| \leq \lambda )
\leq 2 (n_1+n_2) e^{-\frac{1}{2 \lambda^2}}.
$$

Let $A$ and $B$ be the events
$\{\sup_{t_0 \in T^c}|\pi(t_0)|>1\} $
and
$\{\sup_{t_0 \in T^c}\|\omega^0\| \leq \lambda \}$
respectively.

Then
$$
P(A)=P(A|B)P(B)+P(A|B^c)P(B^c) \leq P(A|B)+P(B^c),
$$
which proves the result.
\end{proof}

\noindent
{\bf Proof of Theorem \ref{Th8}.}
Set $\lambda=(\frac{1}{2}+\|\Phi_{T_1}^*\Psi_{T_2} /N\|)^{-1} (\mu (\Phi, \Psi) \sqrt{|T|}+ \frac{a\bar{\sigma}}{m}+ +\frac{ \sqrt{|T|} \mu_M}{\sqrt{m}})$.
Combining Lemma \ref{Lemma10} and \ref{Lemma11},
we have for $\bar{\sigma}$ and each $a>0$ satisfying the conditions in Lemma  \ref{Lemma9},
\begin{eqnarray*}
P(\sup_{t \in T^c} |\pi(t)| \geq 1 )
&\leq& 2 (n_1+n_2) e^{-\frac{1}{2 \lambda^2}}
+ 3(n_1+n_2) e^{-\gamma a^2}  \\
&+& P\left(\|[\Phi_{\Omega T_1} \quad \Psi_{\Omega T_2}]^* [\Phi_{\Omega T_1} \quad \Psi_{\Omega T_2}]\|
\leq (\frac{1}{2}+\|\Phi_{T_1}^*\Psi_{T_2} /N\|) m \right)
\end{eqnarray*}
\hfill $\Box$


\end{document}